\documentclass{article}
\usepackage{spconf,amsmath,graphicx}

\usepackage[utf8]{inputenc}
\usepackage{url}
\usepackage{booktabs}
\usepackage{amsfonts}
\usepackage{microtype}
\usepackage[table]{xcolor}

\usepackage{multicol, multirow}
\usepackage{makecell}
\usepackage{array}

\usepackage{kotex}
\usepackage{pifont}
\usepackage{ctable}

\usepackage{arydshln}
\usepackage{amsthm, amssymb}
\usepackage{algorithm}

\usepackage{listings}

\DeclareMathOperator*{\argmax}{\arg\hspace{-0.07em}\max}

\newcommand{\xmark}{\ding{55}}
\newcommand{\cmark}{\ding{51}}
\newcommand\norm[1]{\left\lVert#1\right\rVert}

\newtheorem{lem}{Lemma}

\newcommand{\boldparagraph}[1]{\noindent\textbf{#1}\enspace}

\usepackage{stfloats}

\usepackage[skip=5pt]{caption}

\title{LoREnc: Low-Rank Encryption for \\ Securing Foundation Models and LoRA Adapters}

\name{
  Beomjin Ahn$^{1\,*}$,
  Jungmin Kwon$^{3\,\dagger}$,
  Chanyong Jung$^{4\,\S}$,
  Jaewook Chung$^{2}$
}

\address{
  $^{1}$Samsung Research,
  $^{2}$Samsung Electronics,
  $^{3}$Amazon Web Services,
  $^{4}$University of Michigan
}

\begin{document}
\maketitle

\begingroup
\renewcommand{\thefootnote}{\fnsymbol{footnote}}
\footnotetext[1]{Corresponding author. E-mail: \texttt{beomjin.ahn@samsung.com}}
\footnotetext[2]{This work was completed before the author joined Amazon.}
\footnotetext[4]{Work done while at Samsung Research.}
\endgroup

\begin{abstract}
  Foundation models and low-rank adapters enable efficient on-device generative AI but raise risks such as intellectual property leakage and model recovery attacks. Existing defenses are often impractical because they require retraining or access to the original dataset. We propose LoREnc, a training-free framework that secures both FMs and adapters via spectral truncation and compensation. LoREnc suppresses dominant low-rank components of FM weights, compensates for the missing information in authorized adapters, and further applies orthogonal reparameterization to obscure structural fingerprints of the protected adapter. Unauthorized users produce structurally collapsed outputs, while authorized users recover exact performance. Experiments demonstrate that LoREnc provides strong protection against model recovery with under 1\% computational overhead.
\end{abstract}
\begin{keywords}
  Generative AI, Foundation Models, LoRA, Parameter-Efficient Fine-Tuning
\end{keywords}

\begin{figure*}[b]
  \centering
  \includegraphics[width=1.0\linewidth]{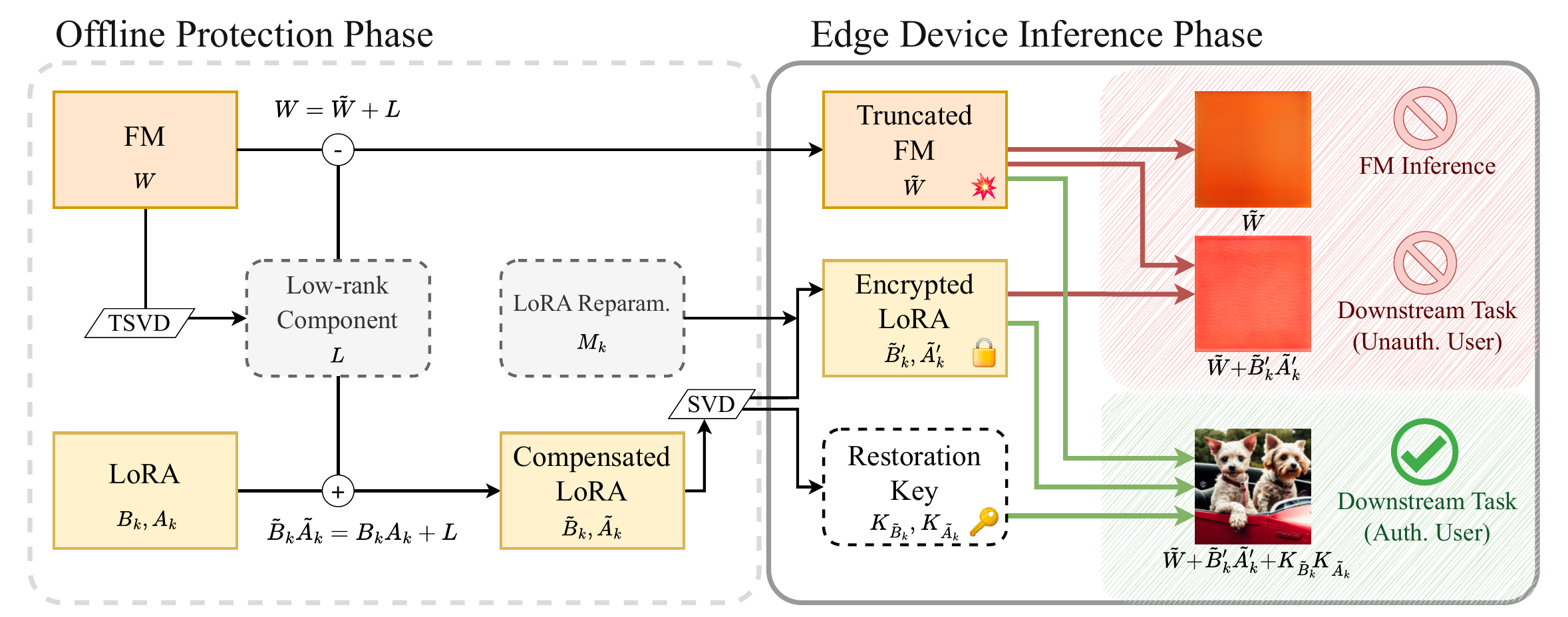}
  \caption{Overview of the LoREnc pipeline. The framework protects FMs by relocating dominant spectral components to LoRA adapters, preventing unauthorized use (visualized as structural collapse) while enabling numerically exact recovery.}
  \label{fig:overall}
\end{figure*}

\section{Introduction}
\label{sec:intro}

Foundation models (FMs) can be adapted to many downstream tasks, improving the practical usability of large-scale models. Parameter-Efficient Fine-Tuning (PEFT) methods are widely adopted for this purpose~\cite{DBLP:journals/tmlr/HanGL0Z24}, and LoRA~\cite{edward2021} is a de facto standard due to its simplicity and broad tooling support. However, releasing FMs also introduces risks: weights can enable unauthorized inference or partial recovery of proprietary models, making exposure especially harmful. Existing protection mechanisms offer limited practical guarantees in this setting. Passive approaches focus on ownership verification rather than preventing unauthorized use. More recent methods attempt to prevent extraction or misuse by modifying or hiding deployed weights, but typically require expensive retraining or still assume reversible parameters are deployed to edge devices. Full-model encryption is also impractical in this setting: runtime decryption of an entire FM requires loading the plaintext model into device memory at inference time, negating the efficiency constraints that define edge deployment.

To address these limitations, we propose LoREnc (Low-Rank Encryption), a training-free framework that jointly protects FMs and their LoRA adapters (Figure~\ref{fig:overall}). Unlike conventional cryptographic methods that secure data confidentiality at the bit level, LoREnc can be interpreted as operating in the spirit of perceptual encryption~\cite{DBLP:journals/tcsv/LiCCBL07}, where unauthorized access leads to severe semantic degradation of model outputs, and the protection is realized directly in the model's weight space. Inspired by the Eckart--Young theorem~\cite{eckart1936approximation}, LoREnc mathematically suppresses the dominant low-rank components of FM weights to structurally degrade unauthorized inference outputs. Conversely, it compensates for these components in authorized adapters to enable theoretically exact recovery of original performance. Unlike prior approaches, LoREnc operates purely on post-training weights without accessing the original dataset, thereby ensuring \emph{data-independence} suitable for privacy-sensitive on-device deployment. Specifically, we propose a training-free spectral truncation and compensation mechanism that preserves authorized performance while inducing structural collapse for unauthorized users. We further introduce a secure adapter encoding scheme robust against reuse and recovery attacks. Extensive experiments, including on-device benchmarks, confirm that LoREnc achieves strong protection with under 1\% overhead.

\section{Related Work}
\label{sec:related_work}

\subsection{Vulnerabilities in Edge Deployment}
Deploying deep learning models on edge devices exposes model weights to adversaries with physical or software-level access, making unauthorized reuse, extraction, and model stealing practical at scale~\cite{sun2021mind, xu2019first, ren2024demistify, deepsteal, huang2022smart}.
Moreover, PEFT and lightweight adapters such as LoRA~\cite{edward2021} simplify edge deployment, but can also facilitate attacks by providing structured update signals. For example, Spectral DeTuning~\cite{horwitz2024recovering} shows that collecting merged FM and adapter weights can recover pre-trained parameters via iterative low-rank factorization.

\subsection{Model Protection and Encryption}
Model protection approaches can be broadly categorized into passive and active methods. Unlike passive techniques such as watermarking and fingerprinting~\cite{zhang2018protecting, yang2021robust}, active methods restrict the model's functionality.
Representative active methods hide important layers in secure storage (e.g., SOTER~\cite{DBLP:conf/usenix/ShenQJWWCZWCLZC22}, ShadowNet~\cite{DBLP:conf/sp/SunSLCLJ23}), obfuscate weights (e.g., NNSplitter~\cite{zhou2023nnsplitteractivedefensesolution}, GroupCover~\cite{DBLP:conf/icml/Zhang0ZZZ0W24}), or decompose parameters (e.g., SLIP~\cite{DBLP:journals/corr/abs-2407-10886}) to prevent unauthorized inference or weight extraction. While these provide stronger protection by modifying deployed parameters, they typically rely on retraining or iterative optimization (e.g., NNSplitter), or expose transformed weights via interactive secure-resource protocols at inference time (e.g., SLIP). In contrast, LoREnc is fully on-device, training-free, and data-independent.

\begin{table*}[t]
  \centering
  \caption{Summary of design requirements for practical foundation-model (FM) protection.}
  \label{tab:design_requirements}
  {\renewcommand{\arraystretch}{1.05}
    \begin{tabular}{c|l}
      \hline
      \textbf{Requirement} & \textbf{Description}  \\
      \hline
      {Effectiveness} & Unauthorized foundation inference should yield {semantically meaningless outputs}. \\
      \cdashline{1-2}
      {Integrity} & Authorized downstream inference should {exactly match} baseline performance. \\
      \cdashline{1-2}
      {Stealthiness} & Protected weights should not appear structurally distinct from ordinary adapters. \\
      \cdashline{1-2}
      {Efficiency} & {Authorized inference} should incur minimal computational and memory overhead. \\
      \cdashline{1-2}
      {Resilience} & The FM should remain unrecoverable under model separation and restoration attacks. \\
      \cdashline{1-2}
      {Data-independence} & No training data should be required for encryption or downstream authorization. \\
      \hline
  \end{tabular}}
\end{table*}

\section{Problem Definition and Threat Model}
\label{sec:problem_definition}

Our objective is to protect the deployed FM weights against unauthorized reuse while preserving the functionality of authorized downstream tasks using LoRA adapters. To this end, we consider a training-free protection setting in which subsets of model parameters are secured and distributed with LoRA adapters, thereby allowing only authorized users to recover the intended behavior.

\subsection{Threat Model and Assumptions}
Unlike server-side deployments, on-device models reside in user-controlled environments where physical memory inspection and static weight analysis are readily available. We assume restoration keys are protected in a hardware-backed environment such as a Trusted Execution Environment (TEE), while the deployed artifacts (encrypted FM weights and encrypted adapters) are accessible to an unauthorized party. The adversary then attempts restoration via ML-level weight-extraction methods such as Spectral DeTuning (SDT)~\cite{horwitz2024recovering} or limited fine-tuning. LoREnc targets practical empirical resistance against such ML-level extraction, rather than formal cryptographic unrecoverability; physical side-channel attacks and direct key leakage are outside the scope of this work.

\subsection{Design Requirements}

We define six design requirements for practical FM protection, summarized in Table~\ref{tab:design_requirements}, which serve as evaluation criteria throughout this paper.
The first five requirements are adopted from prior work~\cite{zhou2023nnsplitteractivedefensesolution}, and we introduce \emph{data-independence} to reflect a realistic situation in which collecting training data and retraining models become impractical.

\vspace{-0.3cm}
\begin{table*}[t!]
  \setlength\tabcolsep{0pt}
  \caption{Visualization of text-to-image results with SD 1.5. The first row shows the baseline results; the remaining rows depict outputs from LoREnc. (Prompt: ``A trio of dogs sitting in their owner's lap in a red convertible.'') \label{tab:vision-visualizations}}
  \begin{center}
    {\renewcommand{\arraystretch}{1.3}
      \begin{tabular}{c:c|c:ccccc}
        \hline

        \multirow{2}{*}{  Model  } & \multirow{2}{*}{  Authorization  } & \multirow{2}{*}{Foundation} & \multicolumn{5}{c}{Downstream} \\
        \cdashline{4-8}
        & & & Task 1 & Task 2 & Task 3 & Task 4 & Task 5 \\

        \hline

        \multicolumn{2}{c|}{\raisebox{3.\height}{Baseline}} &
        \raisebox{-0.1\height}{\includegraphics[width=0.13\linewidth]{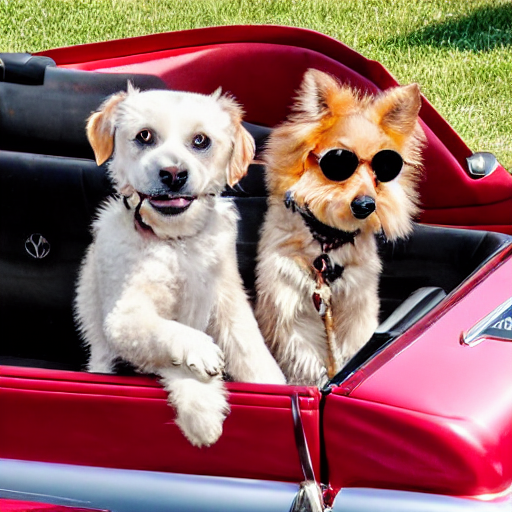}} &
        \raisebox{-0.1\height}{\includegraphics[width=0.13\linewidth]{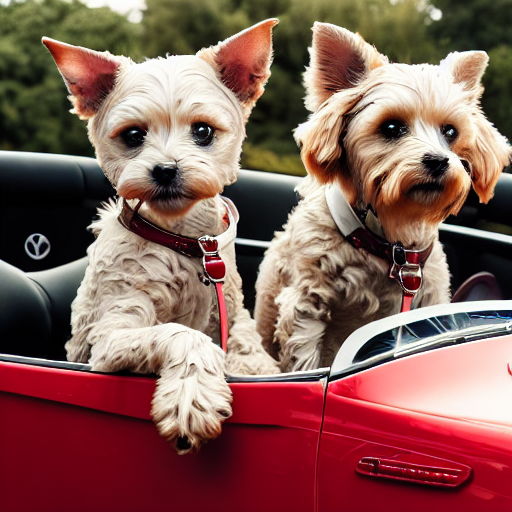}} &
        \raisebox{-0.1\height}{\includegraphics[width=0.13\linewidth]{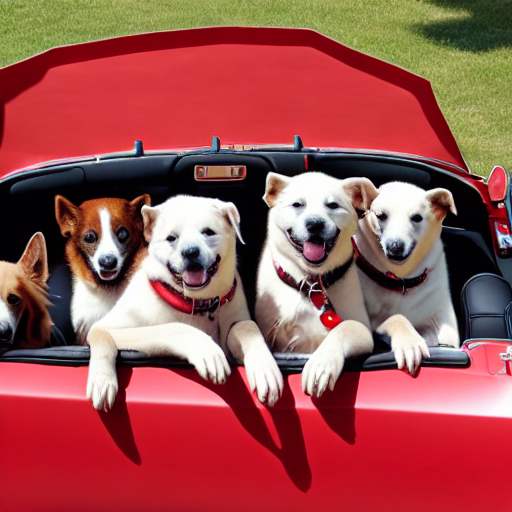}} &
        \raisebox{-0.1\height}{\includegraphics[width=0.13\linewidth]{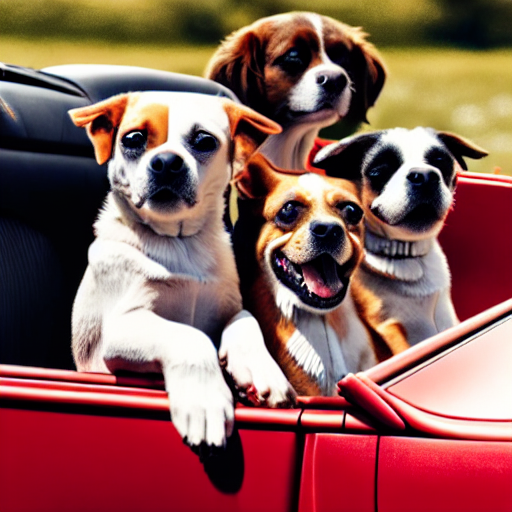}} &
        \raisebox{-0.1\height}{\includegraphics[width=0.13\linewidth]{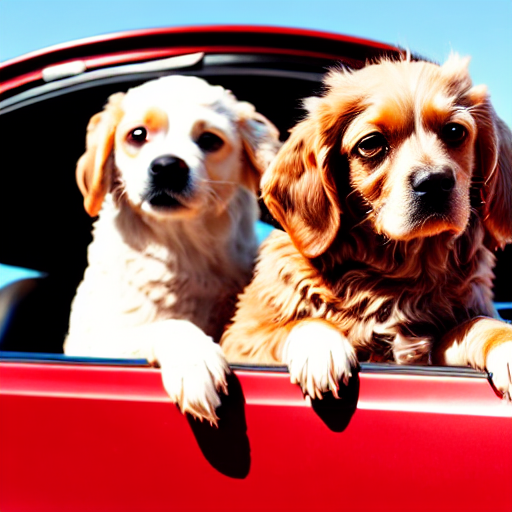}} &
        \raisebox{-0.1\height}{\includegraphics[width=0.13\linewidth]{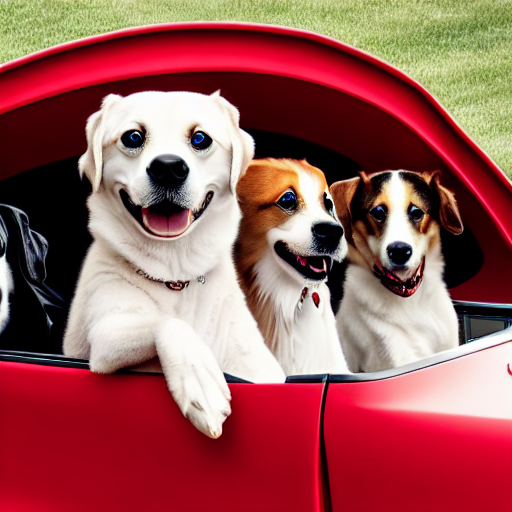}} \\

        \hline

        \multirow{2}{*}{ LoREnc } & \raisebox{3.\height}{\xmark} &
        \raisebox{-0.1\height}{\includegraphics[width=0.13\linewidth]{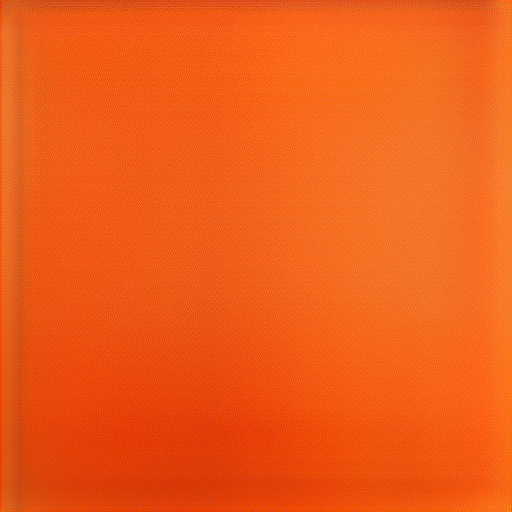}} &
        \raisebox{-0.1\height}{\includegraphics[width=0.13\linewidth]{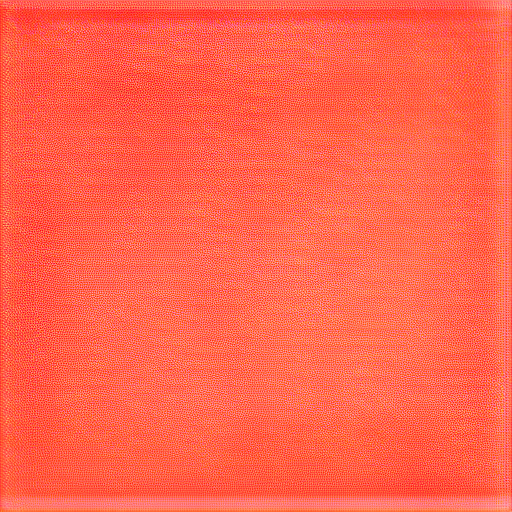}} &
        \raisebox{-0.1\height}{\includegraphics[width=0.13\linewidth]{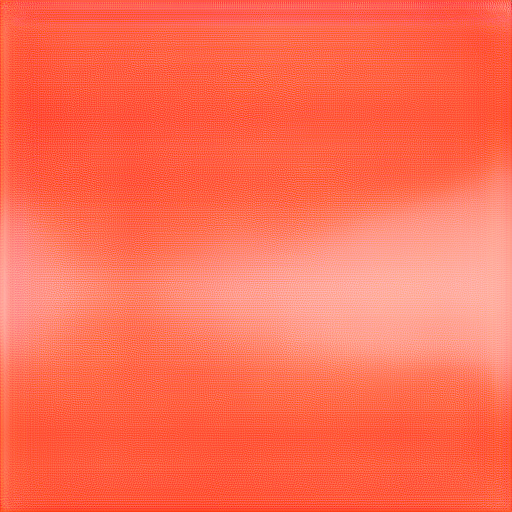}} &
        \raisebox{-0.1\height}{\includegraphics[width=0.13\linewidth]{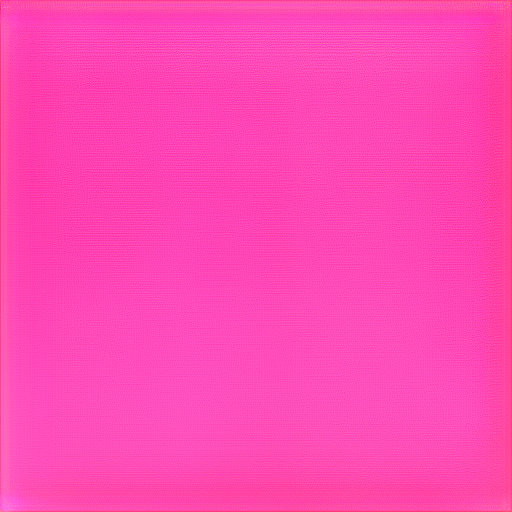}} &
        \raisebox{-0.1\height}{\includegraphics[width=0.13\linewidth]{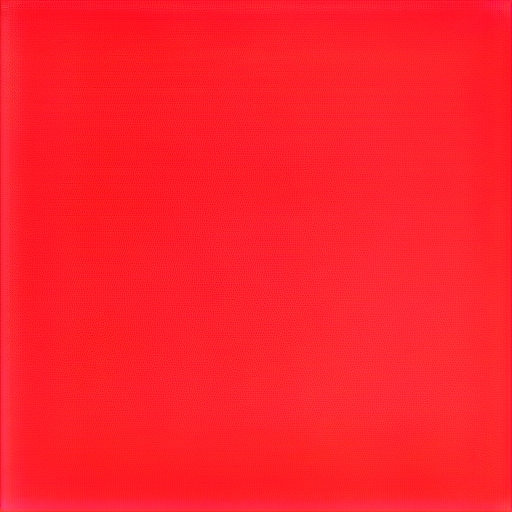}} &
        \raisebox{-0.1\height}{\includegraphics[width=0.13\linewidth]{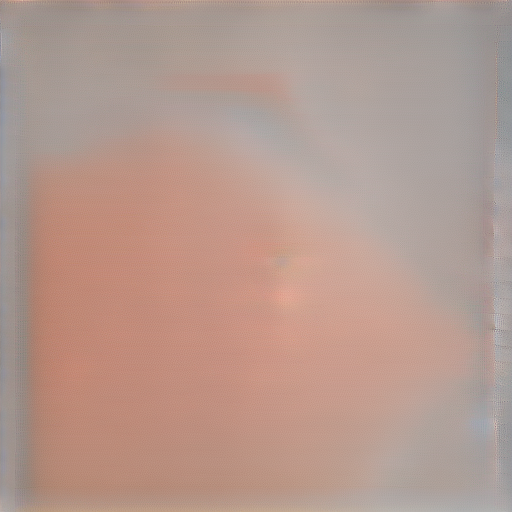}} \\

        \cdashline{2-8}

        & \raisebox{3.\height}{\cmark} &
        \raisebox{-0.1\height}{\includegraphics[width=0.13\linewidth]{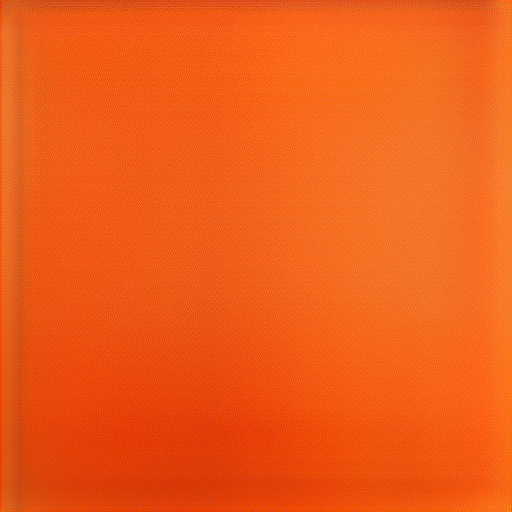}} &
        \raisebox{-0.1\height}{\includegraphics[width=0.13\linewidth]{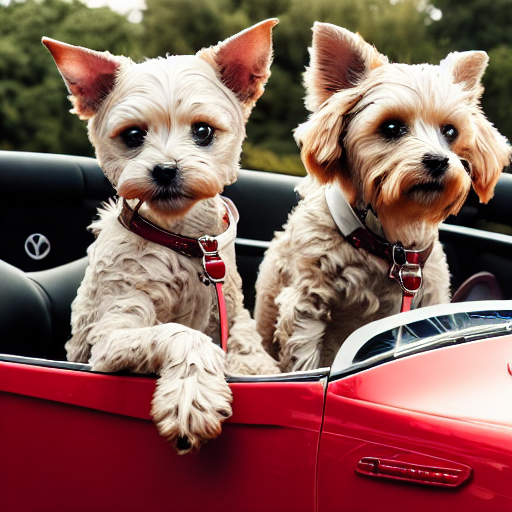}} &
        \raisebox{-0.1\height}{\includegraphics[width=0.13\linewidth]{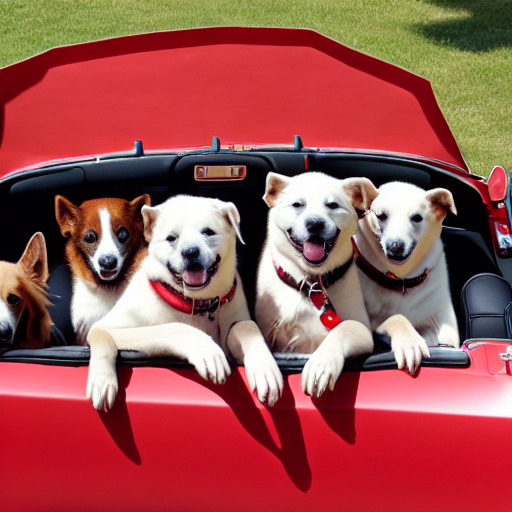}} &
        \raisebox{-0.1\height}{\includegraphics[width=0.13\linewidth]{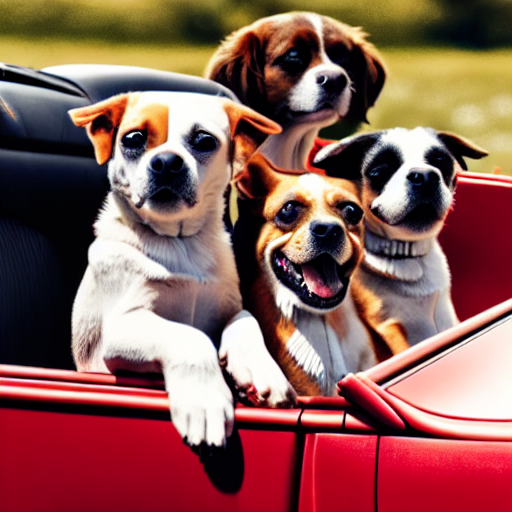}} &
        \raisebox{-0.1\height}{\includegraphics[width=0.13\linewidth]{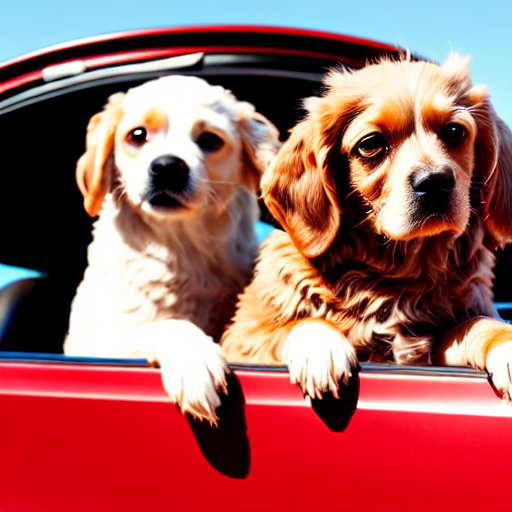}} &
        \raisebox{-0.1\height}{\includegraphics[width=0.13\linewidth]{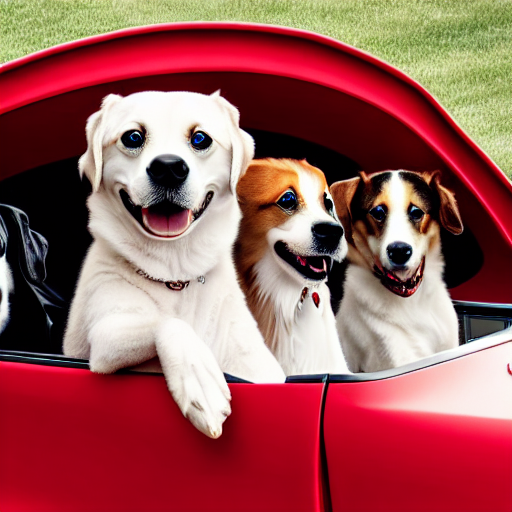}} \\

        \hline

    \end{tabular}}
  \end{center}
\end{table*}

\newcommand{\meanstd}[2]{\makecell{$#1$ \\[-2.5pt]
    \begin{scriptsize}($\pm#2$)
\end{scriptsize}}}

\begin{table*}[h]
  \setlength{\tabcolsep}{12pt}
  \caption{Performance test of LoREnc on SD 1.5 with COCO Captions~\cite{chen2015microsoft}. LPIPS is computed between images generated by the baseline model and the LoREnc-protected model, and $\Delta$CLIP is computed as $\mathrm{CLIP}_{\mathrm{LoREnc}}-\mathrm{CLIP}_{\mathrm{Baseline}}$ for each task. Lower $\Delta$CLIP values indicate more severe performance degradation. For readability, 0.000 values are visually emphasized. \label{tab:vision-metrics}}
  \begin{center}
    \begin{tabular}{c|c|c:ccccc}
      \hline

      & \multirow{2}{*}{\makecell{Authorization}} & \multirow{2}{*}{Foundation} & \multicolumn{5}{c}{Downstream} \\
      \cdashline{4-8}
      &  &  & Task 1 & Task 2 & Task 3 & Task 4 & Task 5 \\
      \hline

      \multirow{2}{*}{$\Delta$CLIP score} & \xmark &
      \meanstd{-0.148}{0.046} &
      \meanstd{-0.155}{0.052} &
      \meanstd{-0.148}{0.048} &
      \meanstd{-0.143}{0.044} &
      \meanstd{-0.144}{0.049} &
      \meanstd{-0.149}{0.048} \\

      \cdashline{2-8}

      & \cmark &
      \meanstd{-0.148}{0.046} &
      \meanstd{\textbf{\textit{0.000}}}{0.000} &
      \meanstd{\textbf{\textit{0.000}}}{0.000} &
      \meanstd{\textbf{\textit{0.000}}}{0.000} &
      \meanstd{\textbf{\textit{0.000}}}{0.001} &
      \meanstd{\textbf{\textit{0.000}}}{0.000} \\

      \hline

      \multirow{2}{*}{LPIPS} & \xmark &
      \meanstd{0.827}{0.080} &
      \meanstd{0.870}{0.074} &
      \meanstd{0.857}{0.073} &
      \meanstd{0.869}{0.084} &
      \meanstd{0.860}{0.084} &
      \meanstd{0.844}{0.085} \\

      \cdashline{2-8}

      & \cmark &
      \meanstd{0.827}{0.080} &
      \meanstd{\textbf{\textit{0.000}}}{0.000} &
      \meanstd{\textbf{\textit{0.000}}}{0.000} &
      \meanstd{\textbf{\textit{0.000}}}{0.000} &
      \meanstd{\textbf{\textit{0.000}}}{0.001} &
      \meanstd{\textbf{\textit{0.000}}}{0.000} \\

      \hline
    \end{tabular}
  \end{center}
\end{table*}

\begin{table*}[h]

  \setlength{\tabcolsep}{6pt}

  \caption{Efficacy on autoregressive FMs evaluated on WikiText-2~\cite{merity2017pointer} under authorized access. (Left) Increase in perplexity ($\Delta$PPL) after applying LoREnc. (Right) Example outputs from the protected FMs (input: ``Kirby 's Block Ball is''). \label{tab:text-results}}
  {\renewcommand{\arraystretch}{1.2}
    \begin{tabular}{c|c|c:c}

      & Model & Foundation & Downstream \\

      \hline

      \multirow{2}{*}{$\Delta$PPL} & GPT-2 &
      $120.0$ &
      $0.000$ \\

      \cline{2-4}

      & Llama 3 &
      $8793$ &
      $0.000$ \\

      \hline
    \end{tabular}
    \qquad
    \begin{tabular}{c|l}
      \multicolumn{2}{c}{Example output of the autoregressive FM (GPT-2)} \\
      \hline
      Baseline &
      {\small {Kirby 's Block Ball is} a special item that can be used ...}\\
      \hdashline
      LoREnc &
      {\small {Kirby 's Block Ball is} a " a " a " a " a " a " a " a " a " ...} \\
      \hline
  \end{tabular}}
\end{table*}

\section{LoREnc: Low-Rank Encryption}

\subsection{Spectral Truncation}
Let $W \in \mathbb{R}^{m \times n}$ denote the weight matrix of an FM layer. Our objective is to construct a truncated weight $\tilde{W}$ that conceals the principal knowledge of $W$ while enabling theoretically exact downstream recovery. We decompose the weight as $W = \tilde{W} + L$, where $L$ is the low-rank component (serving as the spectral key) extracted via truncated SVD.

\boldparagraph{Low-rank Component Extraction}
To maximally suppress the semantic information of $W$, we utilize the Eckart--Young theorem~\cite{eckart1936approximation}, which states that the leading singular components capture the dominant energy of a matrix. Consequently, removing $L$ effectively eliminates the model's ability to form coherent structures, leaving only high-frequency residuals that lack semantic meaning. In the supplementary material, we further prove that this truncation maximizes the Frobenius-norm distance between the original and truncated weights.
We compute the low-rank component via TSVD as:
\vspace{-0.07cm}
\begin{equation}
  L = U_{FM} \Sigma_{FM} V_{FM}^T = \mathrm{TSVD}_{\Delta r}(W),
  \vspace{-0.07cm}
\end{equation}
where $\mathrm{TSVD}_{\Delta r}(\cdot)$ denotes the rank-$\Delta r$ truncated SVD operator. Here, $U_{FM} \in \mathbb{R}^{m \times \Delta r}$, $\Sigma_{FM} \in \mathbb{R}^{\Delta r \times \Delta r}$, and $V_{FM} \in \mathbb{R}^{n \times \Delta r}$. The hyperparameter $\Delta r$ specifies the number of truncated singular components and thus controls the strength of the perceptual encryption. Increasing $\Delta r$ generally improves security, but comes with a trade-off of higher overhead. Since $L$ is never deployed to the edge device, reconstruction of $W$ from $\tilde{W}$ alone is infeasible.

\boldparagraph{Spectral Compensation via LoRA}
To preserve downstream functionality, we require the compensated adapters to satisfy $\tilde{W} + \tilde{B}_k \tilde{A}_k = W + B_k A_k$, which yields the condition $\tilde{B}_k \tilde{A}_k = L + B_k A_k$. To guarantee exact compensation of $L$, we employ a temporary rank expansion via concatenation:
\vspace{-0.07cm}
\begin{equation}
  \tilde{B}_k = [\, B_k,\; U_{FM} \Sigma_{FM}^{1/2} \,], \quad
  \tilde{A}_k = [\, A_k,\; \Sigma_{FM}^{1/2} V_{FM}^T \,],
  \vspace{-0.07cm}
\end{equation}
where $\tilde{B}_k \in \mathbb{R}^{m \times (r + \Delta r)}$ and $\tilde{A}_k \in \mathbb{R}^{(r + \Delta r) \times n}$. This construction ensures exact downstream recovery while effectively fusing the low-rank component into the LoRA adapters, satisfying the \emph{integrity} requirement.

\subsection{LoRA Adapter Encryption}
Since the compensated adapters ($\tilde{B}_k, \tilde{A}_k$) explicitly contain the spectral key $L$, unauthorized access could compromise both the adapter and the foundation model. We therefore introduce an explicit LoRA adapter encryption stage to protect LoRA modules against unauthorized access.

\boldparagraph{LoRA Restoration Keys}
We apply SVD to the adapter weights $\tilde{B}_k \tilde{A}_k = U_{Lo} \Sigma_{Lo} V_{Lo}^T$ and split it as
\vspace{-0.07cm}
\begin{equation}
  [K_{\tilde{B}_k}, \tilde{B}^*_k] = U_{Lo} \Sigma_{Lo}^{1/2}, \quad
  [K_{\tilde{A}_k}, \tilde{A}^*_k] = \Sigma_{Lo}^{1/2} V_{Lo}^T.
  \vspace{-0.07cm}
\end{equation}

Here, $\tilde{B}^*_k \in \mathbb{R}^{m \times r}$ and $\tilde{A}^*_k \in \mathbb{R}^{r \times n}$ denote the encrypted LoRA adapter weights, while $K_{\tilde{B}_k}$ and $K_{\tilde{A}_k}$ form the LoRA restoration keys. Beyond encrypting the adapter contents, this step also reduces the LoRA rank from $r+\Delta r$ back to $r$, which helps conceal whether LoREnc has been applied.

\boldparagraph{Orthogonal LoRA Reparameterization}
Finally, we apply a reparameterization
$\tilde{B}'_k = \tilde{B}^*_k M_k, \quad \tilde{A}'_k = M_k^{T} \tilde{A}^*_k$,
with a random orthogonal matrix $M_k \in \mathbb{R}^{r \times r}$.
This induces an isometric rotation in the parameter space, creating infinite equivalent factorizations for the same product. Without this reparameterization, the encrypted adapters would retain the strict orthogonality inherent to SVD, making them distinguishable from standard Gaussian-initialized weights. This structural fingerprint would allow adversaries to easily detect the presence of the protection, thereby compromising the \emph{stealthiness} requirement against simple structural inspection. We note that adaptive detectors specifically designed for protected adapters may still distinguish them, which is outside the scope of this work.

\subsection{Authorized Downstream Inference}
An authorized user retrieves $\tilde{W}$, $\tilde{B}'_k$, and $\tilde{A}'_k$ from storage and obtains the restoration keys $K_{\tilde{B}_k}$ and $K_{\tilde{A}_k}$ from a secure environment. The decrypted weight is obtained as
\vspace{-0.07cm}
\begin{equation}
  \tilde{W} + \tilde{B}'_k \tilde{A}'_k + K_{\tilde{B}_k} K_{\tilde{A}_k}
  = W + B_k A_k.
  \vspace{-0.07cm}
\end{equation}
This reconstruction occurs on-the-fly during the forward pass, requiring no additional memory storage for the restored FM weights.
Notably, even authorized users cannot directly access the original FM weight $W$, as the low-rank component $L$ is never deployed to the device.

\section{Experiments}
We evaluate LoREnc across diverse generative architectures.
To ensure a direct comparison with the state-of-the-art weight-recovery method, Spectral DeTuning~\cite{horwitz2024recovering}, we primarily utilize Stable Diffusion v1.5 (SD 1.5)~\cite{rombach2022high} as our main testbed.
Additionally, we demonstrate the architecture-agnostic nature of LoREnc by providing results on recent DiT-based models (e.g., Sana) in the supplementary material.
Specifically, our experiments address: efficacy of authorized recovery vs.\ unauthorized degradation (Q1), resilience to fine-tuning attacks (Q2), robustness to Spectral DeTuning~\cite{horwitz2024recovering} (Q3), and edge-device efficiency (Q4).
Unless otherwise specified, we set $\Delta r = 4$, as it offers a practical trade-off between effectiveness and computational overhead. Additional details are provided in the supplementary material.

\subsection{Efficacy of Applying LoREnc (Q1)}

We compare three cases: (i) the original model without LoREnc, (ii) LoREnc-applied model under unauthorized access, and (iii) LoREnc-applied model with valid keys. Table~\ref{tab:vision-metrics} reports CLIP~\cite{radford2021learning} and LPIPS~\cite{zhang2018perceptual} scores on SD 1.5~\cite{rombach2022high}. With LoREnc, foundation-only inference is severely degraded, demonstrating strong \emph{effectiveness} against unauthorized access. Conversely, authorized users recover baseline outputs up to negligible floating-point errors, confirming the \emph{integrity} of the downstream tasks. Table~\ref{tab:vision-visualizations} shows structurally collapsed unauthorized outputs and indistinguishable authorized outputs.

Similar trends hold for autoregressive models (Table~\ref{tab:text-results}). These results suggest that this spectral degradation is modality-agnostic. LoREnc successfully induces high perplexity on these models (GPT-2~\cite{radford2019language}, Llama 3~\cite{llama3}), confirming that our method is applicable beyond computer vision.

\begin{table*}[h]

  \setlength\tabcolsep{5pt}

  \caption{Fine-tuning attack resilience on SD 1.5. CLIP scores are measured after one epoch of fine-tuning with varying data sizes. Baseline CLIP score is $0.267$. \label{tab:fine-tuning}}

  \begin{center}

    {\renewcommand{\arraystretch}{1.2}
      \begin{tabular}{c|c:c|c:cccc}

        \hline

        & Method & Training-free & Protected & 0.1k Data & 1k Data & 10k Data & 100k Data \\

        \hline

        \multirow{2}{*}{\makecell{CLIP score\\(Foundation)}} & NNSplitter & \xmark &

        $\makecell{0.187\ \scriptstyle(\pm 0.048)}$ &

        $\makecell{0.240\ \scriptstyle(\pm 0.039)}$ &

        $\makecell{0.267\ \scriptstyle(\pm 0.029)}$ &

        $\makecell{0.267\ \scriptstyle(\pm 0.031)}$ &

        $\makecell{0.251\ \scriptstyle(\pm 0.032)}$ \\

        \cline{2-8}

        & LoREnc & \cmark &

        $\makecell{0.118\ \scriptstyle(\pm 0.034)}$ &

        $\makecell{0.137\ \scriptstyle(\pm 0.035)}$ &

        $\makecell{0.159\ \scriptstyle(\pm 0.032)}$ &

        $\makecell{0.211\ \scriptstyle(\pm 0.040)}$ &

        $\makecell{0.231\ \scriptstyle(\pm 0.034)}$ \\

        \hline

    \end{tabular}}

  \end{center}

\end{table*}

\subsection{Fine-Tuning Attack (Q2)}

We evaluate resilience against adaptive recovery by fine-tuning the truncated SD 1.5 model with 0.1k--100k samples. Table~\ref{tab:fine-tuning} shows that LoREnc consistently prevents meaningful FM recovery, as reflected by low CLIP scores, and outperforms NNSplitter~\cite{zhou2023nnsplitteractivedefensesolution} across all data regimes. Even with 100k samples, performance remains far below the original. Comprehensive ablations on $\Delta r \in \{1, 4, 16, 64\}$ covering computational overhead and defense resilience under both fine-tuning and SDT attacks are consolidated in the supplementary material (Tables~\ref{tab:overhead_detailed} and \ref{tab:fine-tuning-delta-r}).

\begin{figure}[t]
  \centering
  \includegraphics[width=1.0\linewidth]{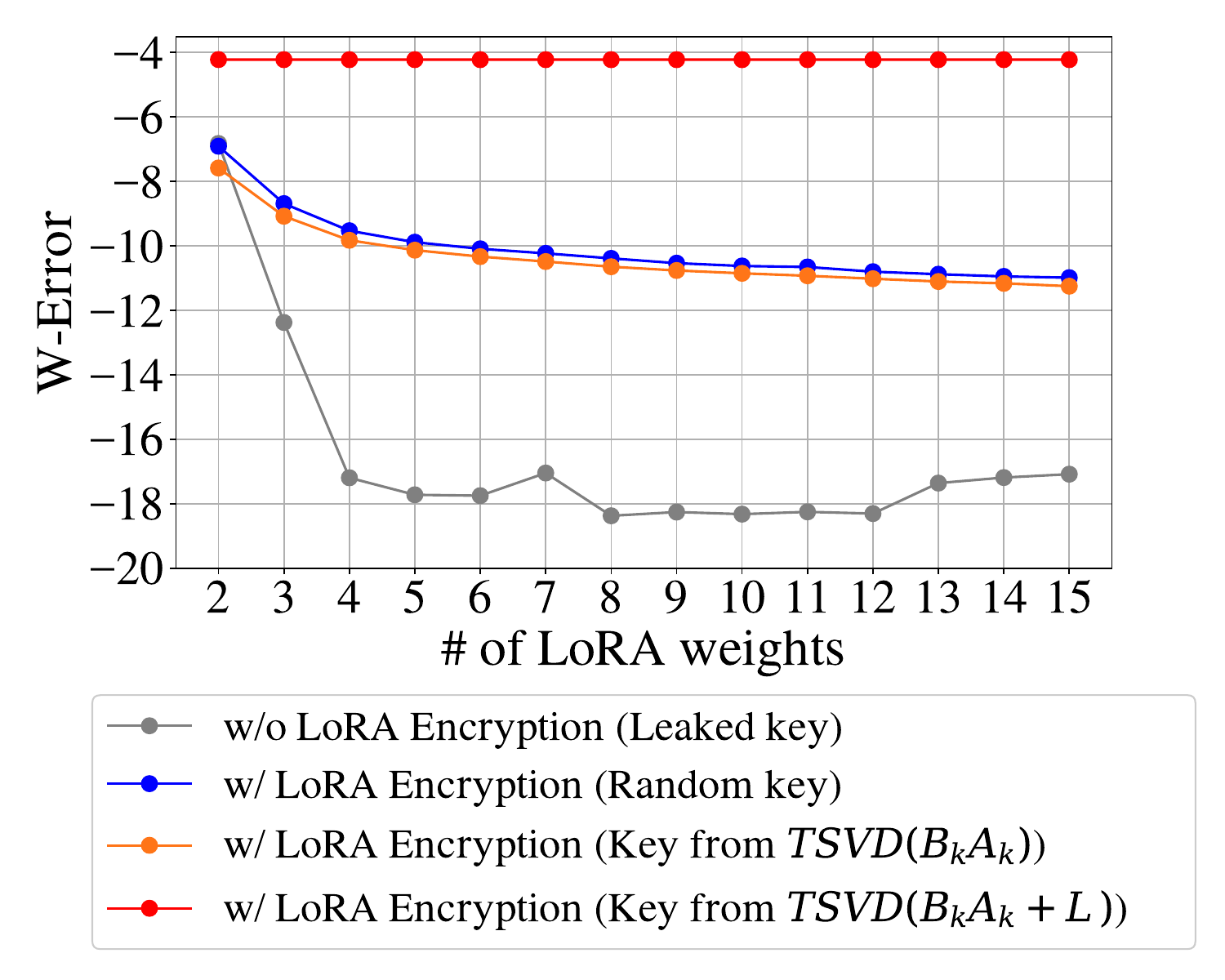}
  \caption{W-Error under Spectral DeTuning with varying numbers of downstream tasks. A higher W-Error indicates better protection.}
  \label{fig:sdt}
\end{figure}

\subsection{Spectral DeTuning Attack (Q3)}

We further evaluate robustness to Spectral DeTuning (SDT)~\cite{horwitz2024recovering}, an advanced attack that attempts to restore FM weights via iterative low-rank factorization. Following the original paper, we measure security using W-Error, the average log-scale mean squared error between estimated and original weights (details in supplementary). Higher W-Error indicates greater deviation from the original weights, implying better protection.
We compare LoREnc against two baseline strategies: (1) random key (blue curve), where the key is generated independently of the weight distribution, and (2) self-derived key (orange curve), where the key is extracted from the adapter itself ($B_k A_k$) without FM context.

As shown in Figure~\ref{fig:sdt}, the random key and the self-derived key fail because they do not capture the structural dependencies and critical spectral information ($L$) of the foundation model. In contrast, LoREnc (red curve) maintains consistently high W-Error, remaining robust under large-scale adapter collection, which highlights the \emph{resilience} of LoREnc against advanced model recovery attempts.

\subsection{\emph{Efficiency} Analysis (Q4)}
We analyze the computational overhead of LoREnc from both theoretical and empirical perspectives.
Theoretically, the computationally expensive operations, such as SVD, are performed entirely before the deployment.
During inference, the only additional cost comes from the increased rank of LoRA adapters ($r \rightarrow r + \Delta r$).
Since $\Delta r$ is typically small (e.g., 4) compared to the hidden dimensions of FMs, this results in a minimal increase in FLOPs and parameter count.
Empirically, as shown in Table~\ref{tab:overhead_result}, LoREnc incurs negligible overheads ($<1\%$) in terms of parameters, GFLOPs, and inference latency measured on a commercial smartphone.
Notably, the latency was evaluated under real-world memory constraints (running the UNet mid-block), confirming that LoREnc is highly suitable for resource-constrained edge environments without requiring specialized hardware accelerators.

\begin{table}[t]
  \caption{Measuring the overhead of LoREnc with UNet of SD 1.5. Inference time is measured by running only the mid-block of the UNet due to real-device memory limitations.} \label{tab:overhead_result}
  \begin{center}
    \begin{small}
      {\renewcommand{\arraystretch}{1.2}
        \begin{tabular}{c|ccc}
          \hline
          & Baseline & LoREnc & Overhead \\
          \hline
          \# of Params (M) & $873.09$ & $874.78$ & $+0.19\%$\\
          \hline
          GFLOPs & $700.18$ & $702.99$ & $+0.40\%$\\
          \hline
          Inference time (sec.) & $0.462$ & $0.463$ & $+0.22\%$\\
          \hline
      \end{tabular}}
    \end{small}
  \end{center}
\end{table}

\section{Conclusion}

We presented LoREnc, a training-free framework employing spectral truncation and compensation to secure on-device FMs. It mathematically guarantees structural collapse for unauthorized inference while preserving integrity for authorized users. In summary, LoREnc satisfies all six design requirements—Effectiveness, Integrity, and Resilience—verified through extensive experiments, while maintaining Stealthiness, Data-independence, and Efficiency essential for practical edge deployment.

\vfill\pagebreak

\clearpage

\setcounter{section}{0}
\setcounter{equation}{0}
\setcounter{table}{0}
\setcounter{figure}{0}

\setlength{\textfloatsep}{20pt plus 2pt minus 4pt}
\setlength{\floatsep}{12pt plus 2pt minus 2pt}
\setlength{\intextsep}{12pt plus 2pt minus 2pt}
\setlength{\dbltextfloatsep}{20pt plus 2pt minus 4pt}
\setlength{\dblfloatsep}{12pt plus 2pt minus 2pt}

\captionsetup{style=base, aboveskip=0pt, belowskip=0pt}
\setlength{\abovecaptionskip}{0pt}
\setlength{\belowcaptionskip}{0pt}
\makeatletter
\long\def\@makecaption#1#2{%
 \vskip 10pt
 \setbox\@tempboxa\hbox{#1. #2}%
 \ifdim \wd\@tempboxa >\hsize #1. #2\par \else \hbox
to\hsize{\hfil\box\@tempboxa\hfil}%
 \fi}
\makeatother

\twocolumn[%
\begin{@twocolumnfalse}
\null
\vskip 2em
\begin{center}
{\large \bf \uppercase{LoREnc: Low-Rank Encryption for \\ Securing Foundation Models and LoRA Adapters\\(Supplemental material)} \par}
\vskip 1.5em
{\large \lineskip .5em
\begin{tabular}[t]{c}{\em }\\ \\  \end{tabular}\par}
\end{center}
\par
\vskip 1.5em
\end{@twocolumnfalse}
]

\renewcommand{\arraystretch}{1.4}

\section{Justification of TSVD-based Truncation}
In the main paper, we claimed that truncating the top-$\Delta r$ singular components maximizes the deviation between the original weights and their truncated counterparts, thereby strengthening our perceptual encryption. We support this claim by deriving the Frobenius norm between the weights.

Let $X \in \mathbb{R}^{m \times n}$ be a real rectangular matrix with $m \geq n$, and let its singular values satisfy $\sigma_1 \geq \cdots \geq \sigma_n \ge 0$.
For an index set $S \subseteq I := \{1,\ldots,n\}$ with $|S|=\Delta r$ (where $1 \leq \Delta r \leq n$), write the SVD expansion as
\begin{equation}
    X = \sum_{i=1}^n \sigma_i \mathbf{u}_i \mathbf{v}^\top_i.
\end{equation}
We define the rank-$\Delta r$ partial sum
\begin{equation} \label{eq:X_S}
    X_S = \sum_{i \in S} \sigma_i \mathbf{u}_i \mathbf{v}^\top_i.
\end{equation}
In particular, define $T := \{1,\ldots,\Delta r\}$, i.e., the indices of the $\Delta r$ largest singular values, and denote
\begin{equation}
    X_T = \sum_{i \in T} \sigma_i \mathbf{u}_i \mathbf{v}^\top_i = \sum_{i=1}^{\Delta r} \sigma_i \mathbf{u}_i \mathbf{v}^\top_i.
\end{equation}

The following lemma formalizes why truncating the most significant singular components maximizes the deviation between the original and truncated weights.
\begin{lem} \label{lemma1}
Among all subsets $S \subseteq I$ with $|S|=\Delta r$, the set $T=\{1,\ldots,\Delta r\}$ maximizes $\norm{X_S}_F$ (equivalently, $\norm{X-(X-X_S)}_F$). That is,
\[
T \in \argmax_{\substack{S \subseteq I \\ |S|=\Delta r}} \norm{X_S}_F.
\]
\end{lem}

\begin{proof}
From Equation~\ref{eq:X_S} and the orthonormality of the singular vectors,
\begin{equation} \label{eq:norm_X_S}
    \norm{X_S}_F^2 = \sum_{i \in S} \sigma_i^2.
\end{equation}
Let $S=\{i_1<\cdots<i_{\Delta r}\}$ be any subset of size $\Delta r$. Since $i_j \ge j$ and the sequence $(\sigma_i)$ is non-increasing, we have $\sigma_{i_j} \le \sigma_j$ for all $j=1,\ldots,\Delta r$. Therefore,
\[
\sum_{i \in S} \sigma_i^2 = \sum_{j=1}^{\Delta r} \sigma_{i_j}^2 \le \sum_{j=1}^{\Delta r} \sigma_j^2 = \sum_{i \in T} \sigma_i^2 = \norm{X_T}_F^2.
\]
Taking square roots yields $\norm{X_S}_F \le \norm{X_T}_F$ for all such $S$, so $T$ is a maximizer.
\end{proof}

\section{Experiment Details} \label{appendix:experiment_details}
Experiments were conducted using an NVIDIA H100 GPU (80GB HBM3), with FP32 precision (w/o NVIDIA TF32).

\subsection{Efficacy of Applying LoREnc (Q1)}
We obtained Stable Diffusion 1.5~\cite{supp:rombach2022high}, GPT-2~\cite{supp:radford2019language}, and Llama 3~\cite{supp:llama3} from Hugging Face (stable-diffusion-v1-5/stable-diffusion-v1-5, openai-community/gpt2, meta-llama/Meta-Llama-3-8B).

For Stable Diffusion, we used five downstream LoRA adapters (ral-bastet-sd15, ral-chrcrts-sd15, ral-cigarette-sd15, ral-cofzee-sd15, ral-crystals-sd15) from the LoWRA Bench dataset~\cite{supp:horwitz2024recovering}. We used the first 100 captions from the COCO Captions validation set~\cite{supp:chen2015microsoft} for text-to-image evaluation (\url{https://github.com/tylin/coco-caption/blob/master/annotations/captions_val2014.json}). CLIP scores were computed using the ViT-L/14 CLIP model~\cite{supp:radford2021learning}, and LPIPS~\cite{supp:zhang2018perceptual} was computed using VGG networks. For generations, we fixed the random seed to $0$ and used the default inference settings.
Table \ref{tab:raw-clip-score} shows the raw results of CLIP scores. $\Delta$CLIP score in the main paper is computed by subtracting the baseline scores (first row) from scores of the LoREnc-applied model (second and third rows).

For GPT-2 and Llama 3, we evaluated on the WikiText-2 test split~\cite{supp:merity2017pointer} (wikitext-2-raw-v1; \url{https://huggingface.co/datasets/wikitext}) and used publicly available LoRA adapters from Hugging Face (varun-v-rao/gpt2-large-lora-2.95M-squad-model1, hallisky/lora-formality-formal-llama-3-8b).

\subsection{Fine-Tuning Attack (Q2)}
Because NNSplitter requires reinforcement-learning-based retraining, we followed the authors' released implementation and hyperparameters. We trained Stable Diffusion 1.5~\cite{supp:rombach2022high} on the LAION dataset~\cite{supp:laion} following the original procedure. For the RNN controller, we used \texttt{epoch=20}, \texttt{batch\_size=5}, and \texttt{lr=0.01}. For Stable Diffusion updates within each RL batch, we used \texttt{batch\_size=8}, \texttt{lr=0.005}, \texttt{min\_w=-0.14}, \texttt{max\_w=0.24}, and \texttt{eps=8e-5} for 12,800 steps. We applied early stopping when the RL reward or the generated images' CLIP score did not improve over several runs. This procedure produced an encrypted model with $2295$ weight secrets.

For the fine-tuning attack experiment, we fine-tuned using AdamW with \texttt{lr=1e-4} and \texttt{batch\_size=32}.

\subsection{Spectral DeTuning Attack (Q3)}
We used the same dataset and experimental protocol as the original Spectral DeTuning paper~\cite{supp:horwitz2024recovering} (\texttt{n\_iters=300}, \texttt{sched\_start\_rank=1}, \texttt{sched\_end\_rank=32}).
Following the paper, we used W-Error defined as
\begin{equation}
  \text{W-Error} = \frac{1}{N} \sum_{l=1}^{N} \log_{10} \left( \frac{1}{|W^{(l)}|} \| \hat{W}^{(l)} - W^{(l)} \|_F^2 \right)
\label{eq:w_error}
\end{equation}
\noindent where $N$ is the total number of layers, $|W^{(l)}|$ is the number of parameters in the $l$-th layer, and $\hat{W}^{(l)}$ denotes the recovered weight.

\subsection{Efficiency Analysis (Q4)}
We used the \texttt{calflops} library (\url{https://github.com/MrYxJ/calculate-flops.pytorch}) to compute GFLOPs for the overhead analysis. For edge-device measurements, we converted the model to LiteRT (formerly TFLite) and tested the official LiteRT Android benchmark application on a Galaxy Fold 4 smartphone (10 warm-up runs; average over 50 measured runs). As described in the main paper, we measured inference time by running only the UNet mid-block due to memory constraints on the device. Detailed parameter and computational overhead results are reported in Table~\ref{tab:overhead_detailed}.

Interestingly, while GFLOPs increase linearly with $\Delta r$, the on-device inference latency shows a non-linear jump at $\Delta r=16$. This behavior is typical in on-device environments, attributed to memory access overheads exceeding cache thresholds or suboptimal kernel tiling for specific matrix dimensions. This observation reinforces our choice of $\Delta r=4$ as the optimal trade-off point for SD 1.5, offering strong security with negligible latency.

\newcommand{\valpct}[2]{\makecell{#1\\(+#2\%)}}
\begin{table*}[t]
\caption{Detailed results of the computation and parameter overhead of LoREnc with various $\Delta r$.} \label{tab:overhead_detailed}
\begin{center}
\begin{small}
\begin{tabular}{c|ccccc}
\hline
 - & Baseline & 1 & 4 & 16 & 64 \\
\Xhline{2pt}
\# of Params (M) & $873.09$ & \valpct{873.51}{0.05} & \valpct{874.78}{0.19}& \valpct{879.87}{0.78}& \valpct{900.22}{3.11}\\
\hline
GFLOPs & $700.18$ & \valpct{700.88}{0.10}& \valpct{702.99}{0.40}& \valpct{711.39}{1.60}& \valpct{745.03}{6.41}\\
\hline
Inference time (Smartphone / sec.) & $0.462$ & \valpct{0.463}{0.22} & \valpct{0.463}{0.22} & \valpct{0.550}{19.0} & \valpct{0.592}{28.1} \\
\hline
\end{tabular}
\end{small}
\end{center}
\end{table*}

\renewcommand{\meanstd}[2]{\makecell{$#1$ \\ \begin{scriptsize}($\pm#2$)\end{scriptsize}}}
\begin{table*}[h!]
\setlength{\tabcolsep}{12pt}
\caption{Raw CLIP scores from the ``Efficacy of Applying LoREnc (Q1)" experiment.}
\label{tab:raw-clip-score}
\begin{center}
\begin{tabular}{c:c|c:ccccc}
\hline

\multirow{2}{*}{Model} & \multirow{2}{*}{Authorization} & \multirow{2}{*}{Foundation} & \multicolumn{5}{c}{Downstream} \\
\cdashline{4-8}
 & & & Task 1 & Task 2 & Task 3 & Task 4 & Task 5 \\
 
\hline

\multicolumn{2}{c|}{\raisebox{0.2\height}{Baseline}} & 
\meanstd{0.267}{0.032} & 
\meanstd{0.271}{0.035} & 
\meanstd{0.270}{0.033} & 
\meanstd{0.260}{0.035} & 
\meanstd{0.263}{0.036} & 
\meanstd{0.265}{0.030} \\

\hline

\multirow{2}{*}{LoREnc} & \raisebox{0.2\height}{\xmark} & 
\meanstd{0.118}{0.034} & 
\meanstd{0.116}{0.032} & 
\meanstd{0.121}{0.031} & 
\meanstd{0.117}{0.027} & 
\meanstd{0.119}{0.029} & 
\meanstd{0.116}{0.032} \\

\cdashline{2-8}

& \raisebox{0.2\height}{\cmark} & 
\meanstd{0.118}{0.034} & 
\meanstd{0.271}{0.035} & 
\meanstd{0.270}{0.033} & 
\meanstd{0.260}{0.035} & 
\meanstd{0.263}{0.036} & 
\meanstd{0.265}{0.030} \\

\hline

\end{tabular}
\end{center}
\end{table*}
\section{Additional Qualitative Results on DiT Architectures}

While our main experiments focus on SD 1.5 for fair comparison with prior baselines, LoREnc is fundamentally a matrix-level operation applicable to any architecture. To verify its generalizability, we evaluate LoREnc on Sana-0.6B~\cite{supp:SANA}, a state-of-the-art Diffusion Transformer (DiT). We specifically selected the model because its extremely efficient footprint makes it a realistic and practical candidate for our targeted on-device edge deployment scenarios. Following standard PEFT practices, we selectively target the attention projection layers (i.e., \texttt{attn.to\_q}, \texttt{q\_proj}, \texttt{to\_k}, \texttt{k\_proj}, \texttt{to\_v}, \texttt{v\_proj}), resulting in exactly 112 protected layers.

Figure~\ref{fig:sana_ablation} visualizes the effect of varying the truncation rank ($\Delta r$). Unlike most large-scale models, Sana-0.6B packs dense spectral information into its weights. Consequently, a minimal spectral truncation ($\Delta r=4$) leads to partial information leakage, and structural concepts (e.g., dog, airplane) remain recognizable. To achieve complete semantic collapse on such dense architectures, a higher truncation rank is required. At $\Delta r=16$, the global structural information is completely lost, leaving only meaningless low-level statistical features (e.g., color). Notably, even with $\Delta r=16$ across all 112 target layers, the parameter overhead added to the LoRA adapters remains negligible at approximately $0.69\%$. This demonstrates that LoREnc securely scales with the model's architectural density while preserving the strict efficiency required for edge devices.

\begin{figure}[t]
    \centering
    
    \begin{minipage}{0.32\linewidth}
        \centering
        \includegraphics[width=\linewidth]{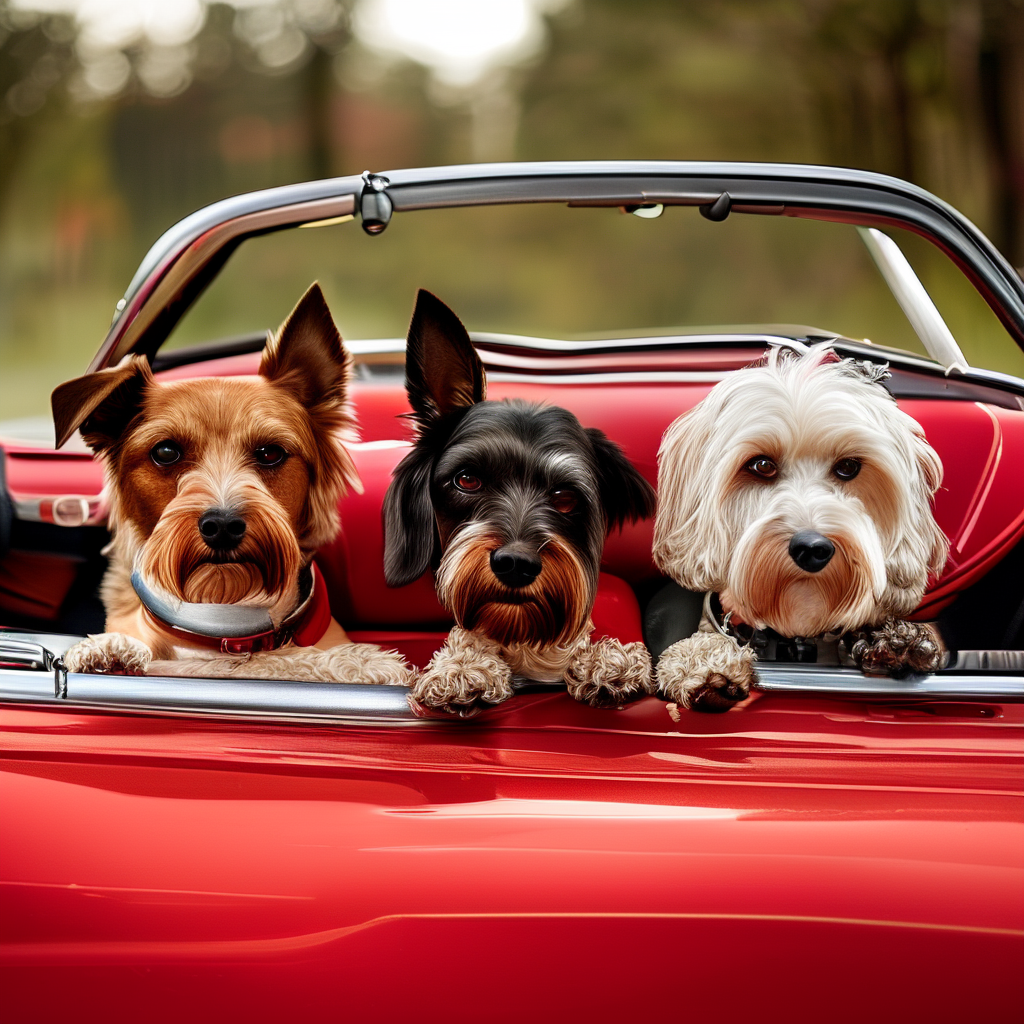}
    \end{minipage}\hfill
    \begin{minipage}{0.32\linewidth}
        \centering
        \includegraphics[width=\linewidth]{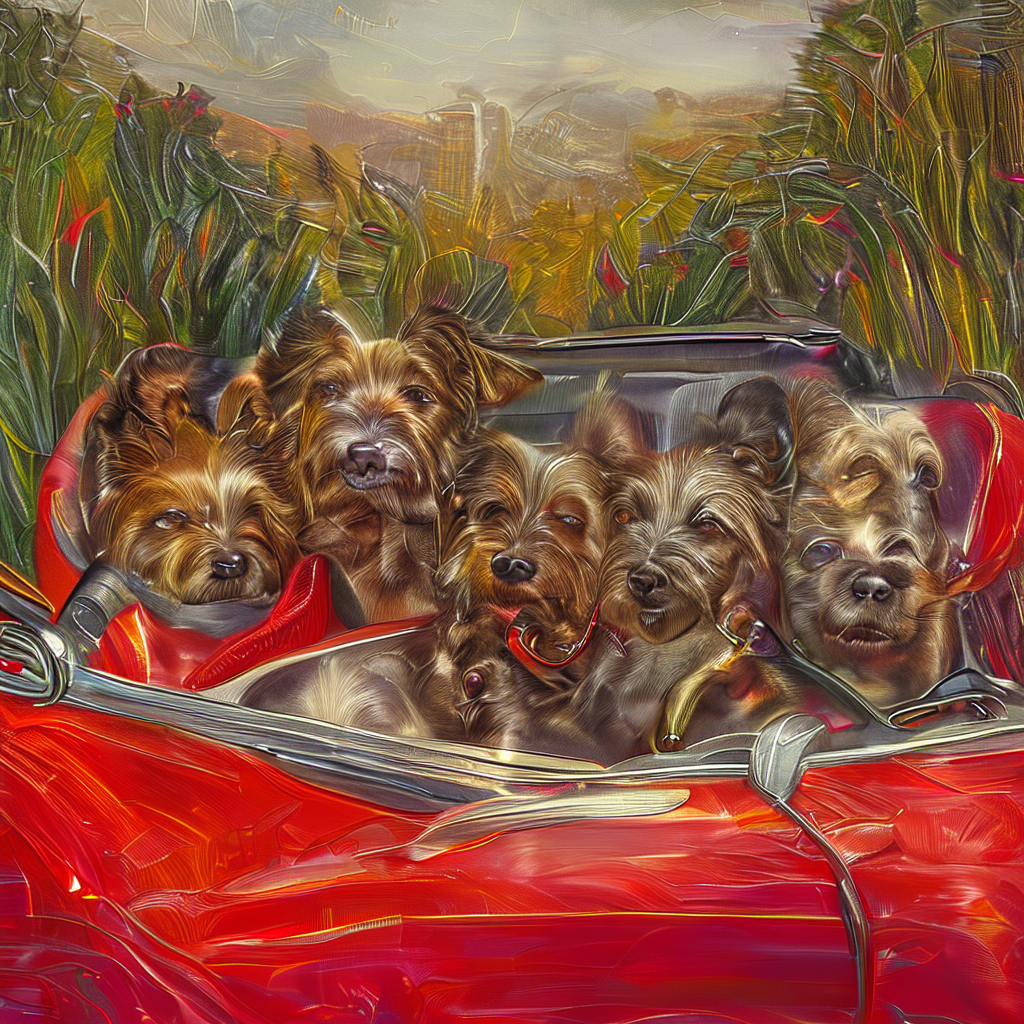}
    \end{minipage}\hfill
    \begin{minipage}{0.32\linewidth}
        \centering
        \includegraphics[width=\linewidth]{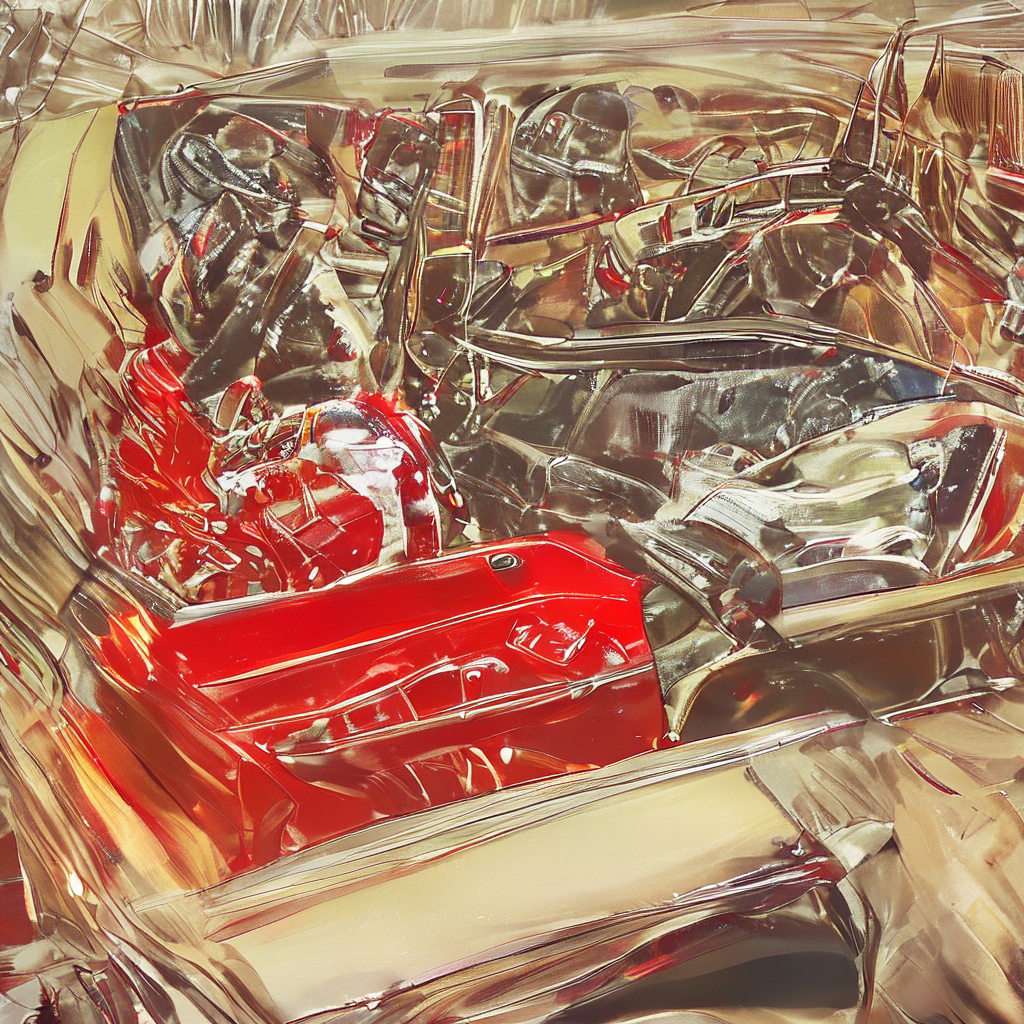}
    \end{minipage}
    
    \vspace{0.15cm}
    
    \begin{minipage}{0.32\linewidth}
        \centering
        \includegraphics[width=\linewidth]{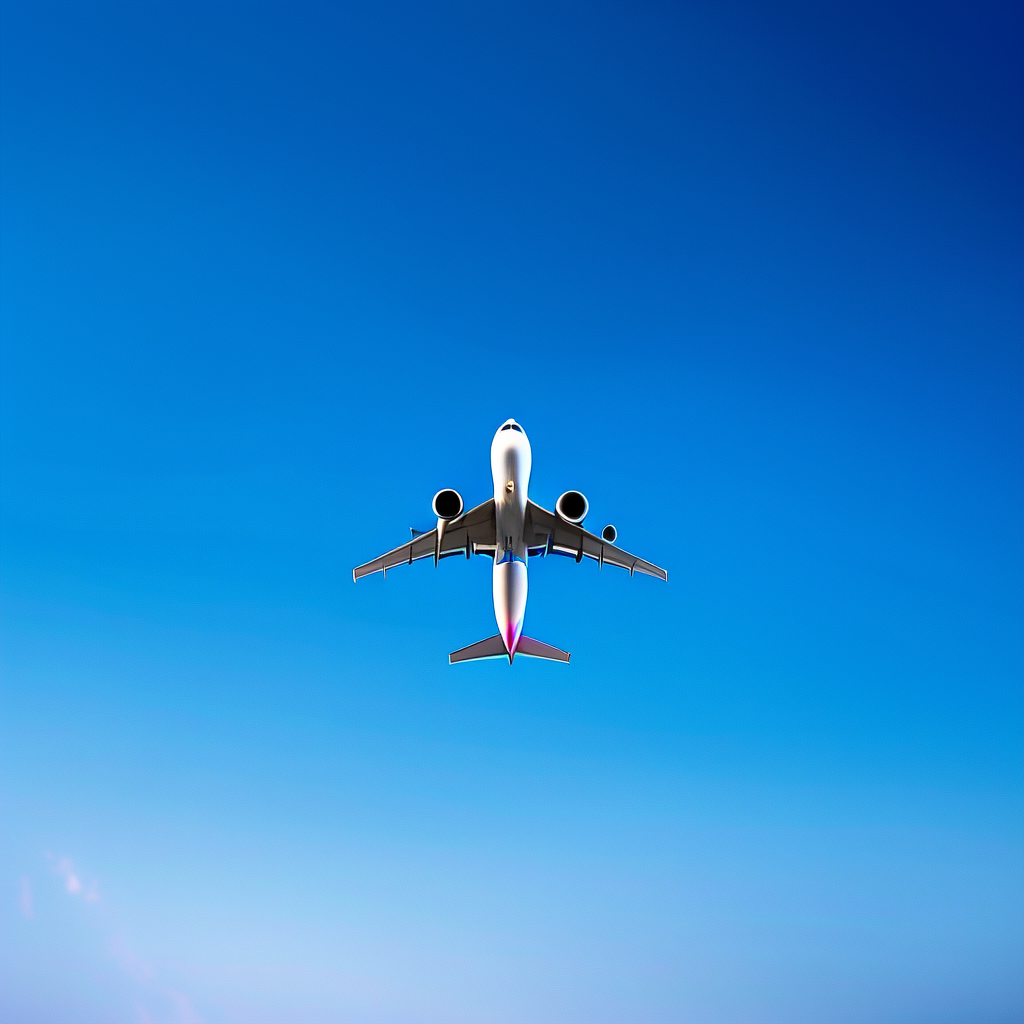}
    \end{minipage}\hfill
    \begin{minipage}{0.32\linewidth}
        \centering
        \includegraphics[width=\linewidth]{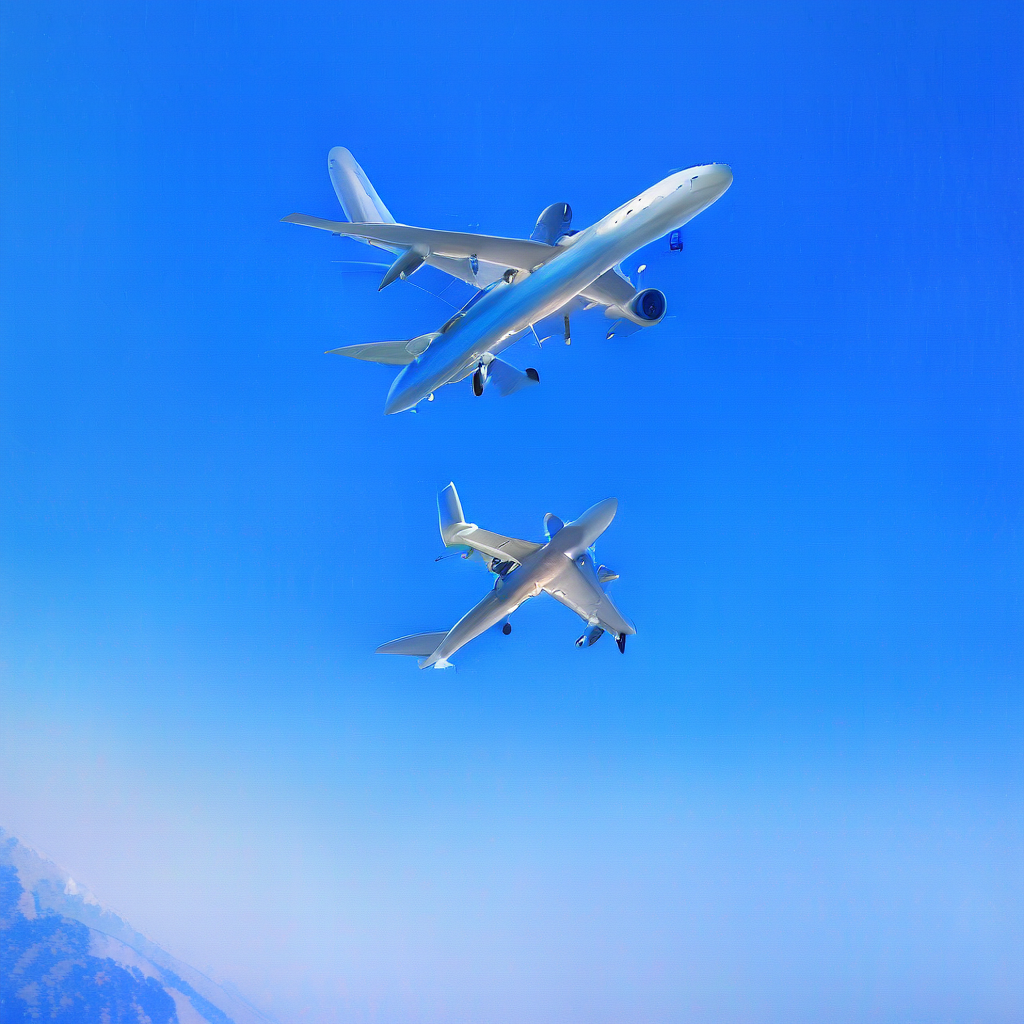}
    \end{minipage}\hfill
    \begin{minipage}{0.32\linewidth}
        \centering
        \includegraphics[width=\linewidth]{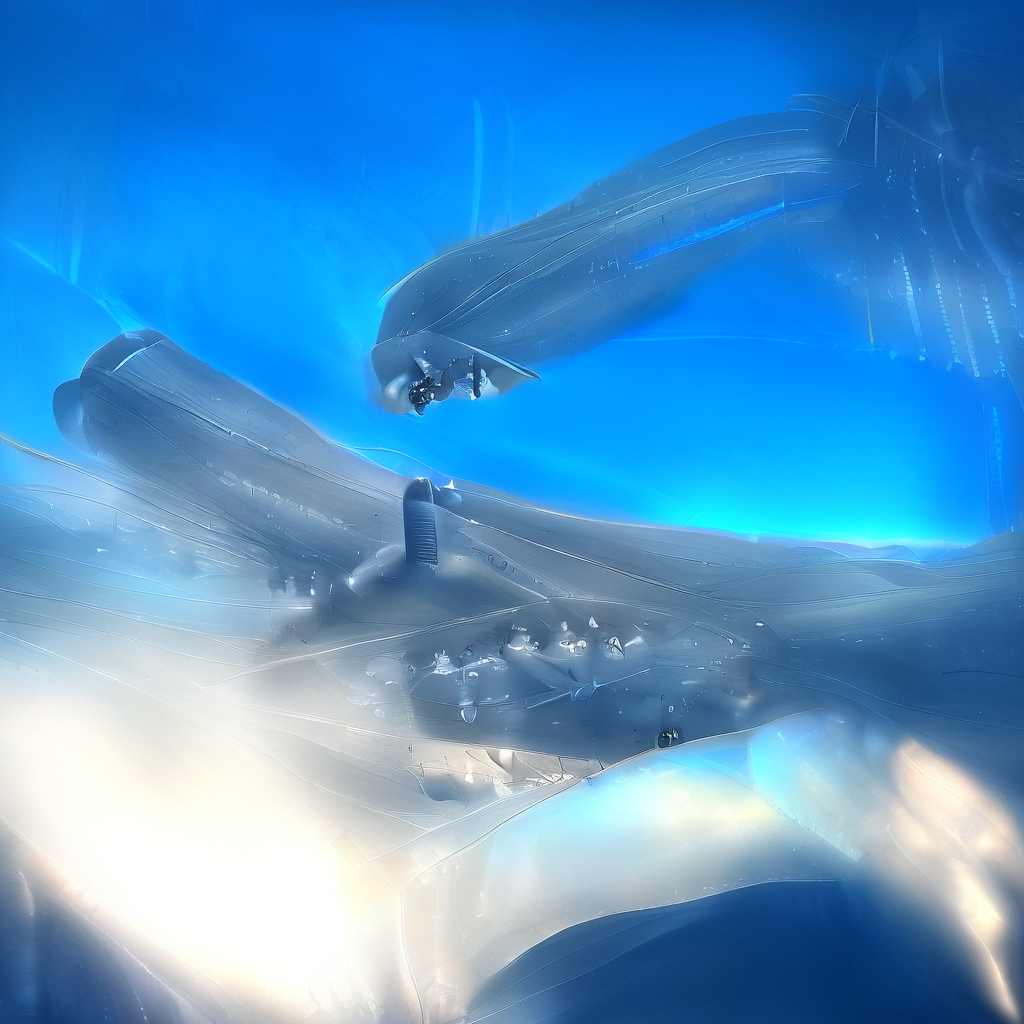}
    \end{minipage}
    
    \vspace{0.15cm}
    
    \begin{minipage}{0.32\linewidth}
        \centering
        \includegraphics[width=\linewidth]{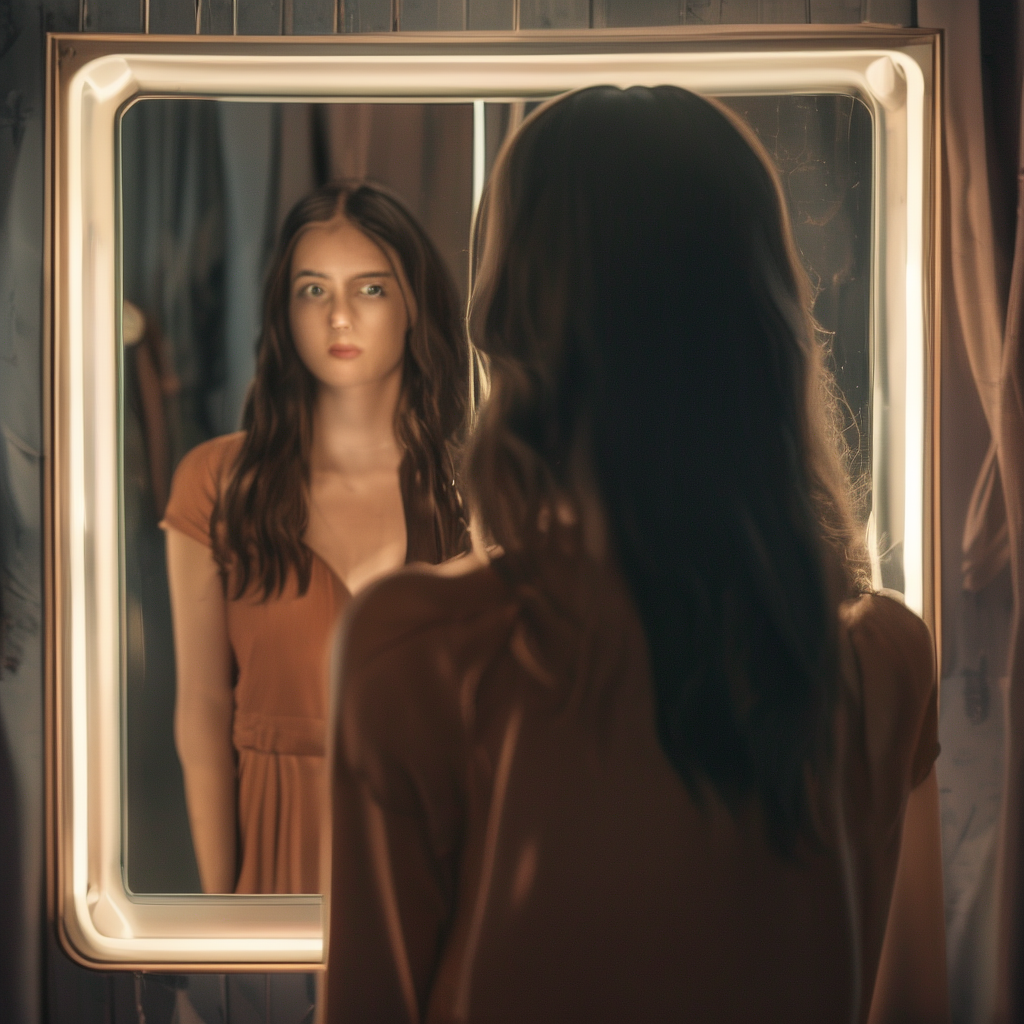}
        \vspace{0.1cm}
        \centerline{\small (a) Original}
    \end{minipage}\hfill
    \begin{minipage}{0.32\linewidth}
        \centering
        \includegraphics[width=\linewidth]{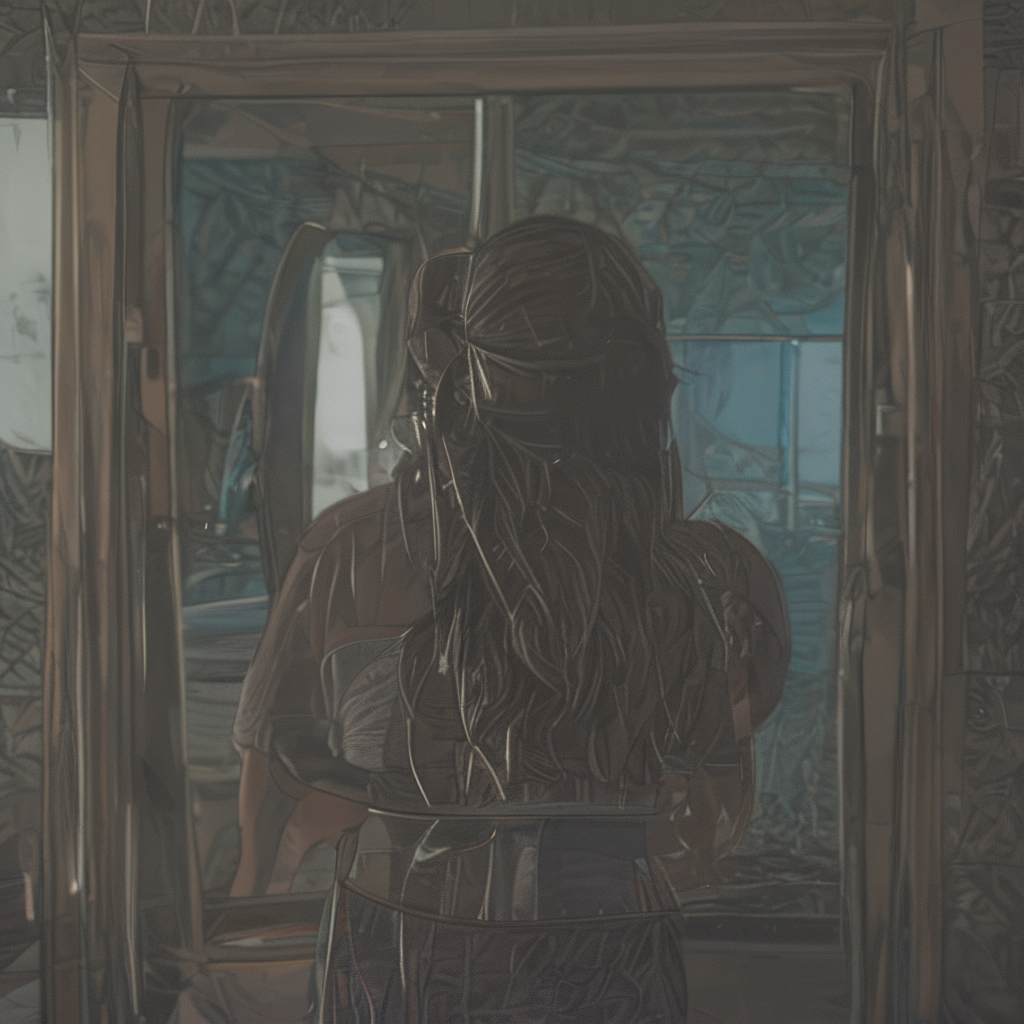}
        \vspace{0.1cm}
        \centerline{\small (b) $\Delta r = 4$}
    \end{minipage}\hfill
    \begin{minipage}{0.32\linewidth}
        \centering
        \includegraphics[width=\linewidth]{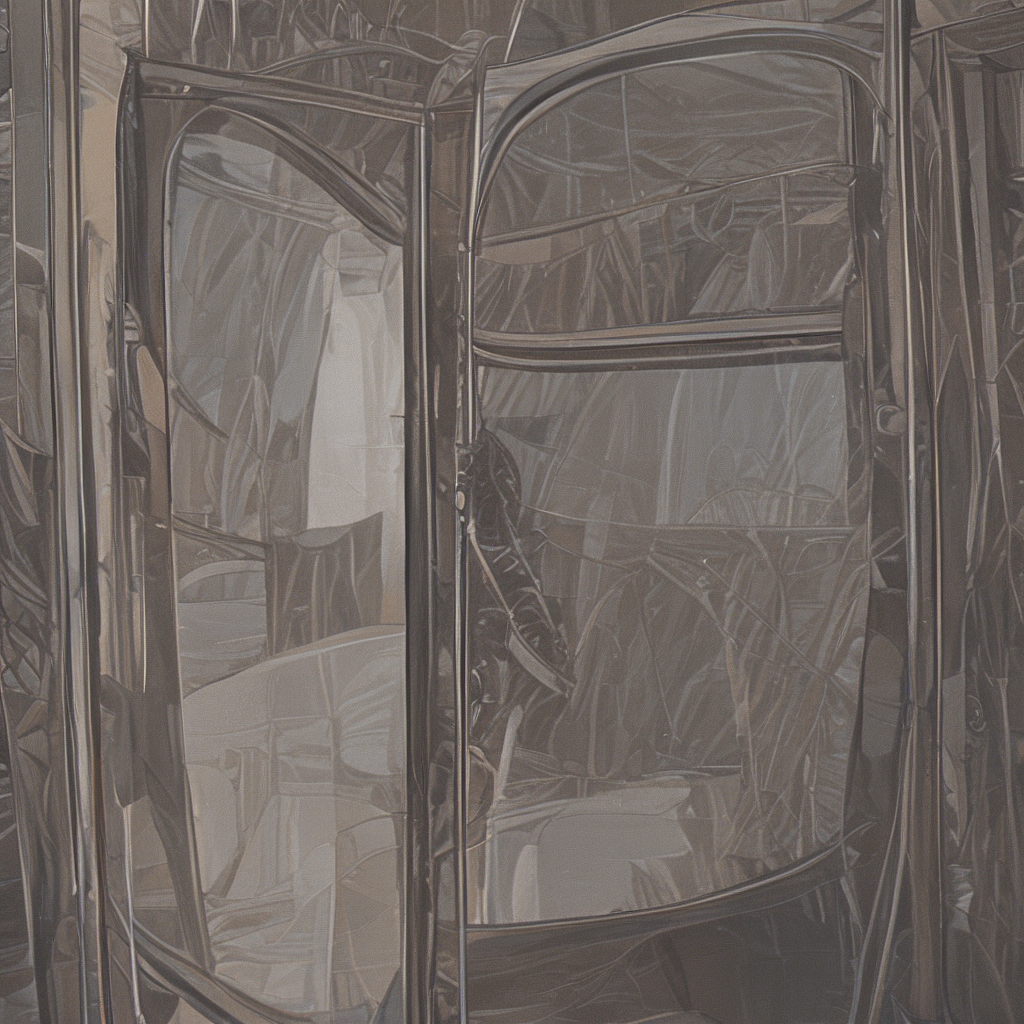}
        \vspace{0.1cm}
        \centerline{\small (c) $\Delta r = 16$}
    \end{minipage}
    
    \vspace{0.2cm}
    \caption{Effect of the truncation rank ($\Delta r$) on a highly compact DiT model (Sana-0.6B). Prompts: ``A trio of dogs sitting in their owner's lap in a red convertible.'', ``An airplane flying high in the blue sky.'', ``A woman stands in front of the mirror to take a picture.''}
    \label{fig:sana_ablation}
\end{figure}

\section{Effect of Varying the $\Delta r$ on Fine-Tuning Attack}

This section reports additional quantitative results and visualizations for the ``Fine-Tuning Attack (Q2)" experiment (Table~\ref{tab:fine-tuning-delta-r}). We further vary $\Delta r$ to illustrate how the truncation strength affects recoverability under fine-tuning. CLIP scores are measured after one epoch of fine-tuning with varying dataset sizes. We also report the corresponding parameter and computational overhead for each $\Delta r$. Unless otherwise stated, we use $\Delta r = 4$.

\begin{table*}[!ht]
\setlength\tabcolsep{4.5pt}
\caption{Fine-tuning attack resilience with varying the $\Delta r$ on Stable Diffusion. The last row shows the result of baseline Stable Diffusion for comparison. (Prompt: ``A trio of dogs sitting in their owner's lap in a red convertible.'') \label{tab:fine-tuning-delta-r}}
\begin{center}
\begin{tabular}{c|c:c:c:c:c}
\hline
\multirow{2}{*}{$\Delta r$} & \multicolumn{5}{c|}{CLIP score}\\
\cdashline{2-6}
 & Protected & 0.1k Data & 1k Data & 10k Data & 100k Data\\
\hline
\multirow{2}{*}{\raisebox{-5.5\height}{1}} &
$\makecell{0.125\ \scriptstyle(\pm 0.031)}$ &
$\makecell{0.146\ \scriptstyle(\pm 0.031)}$ &
$\makecell{0.193\ \scriptstyle(\pm 0.039)}$ &
$\makecell{0.242\ \scriptstyle(\pm 0.032)}$ &
$\makecell{0.249\ \scriptstyle(\pm 0.033)}$ \\
 &
\includegraphics[width=0.13\linewidth]{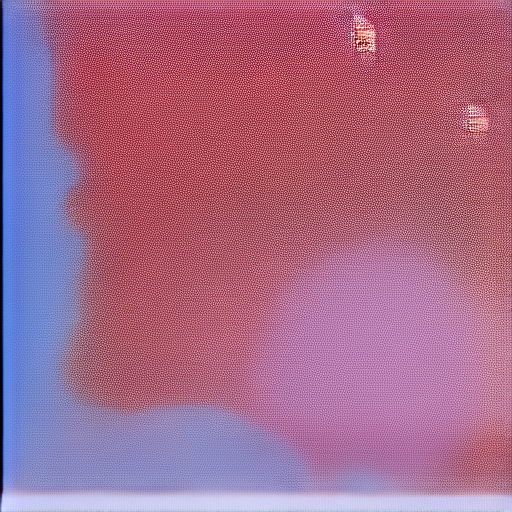} &
\includegraphics[width=0.13\linewidth]{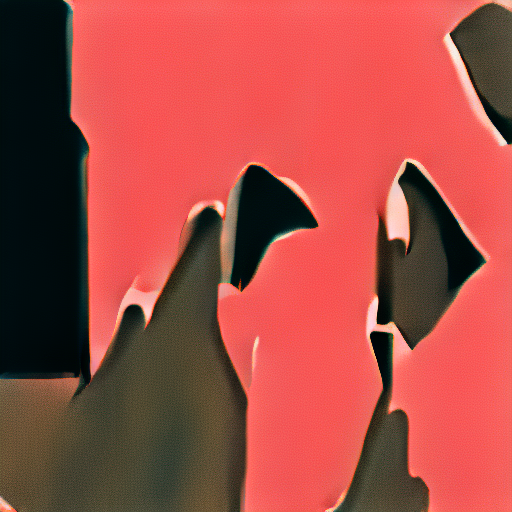} &
\includegraphics[width=0.13\linewidth]{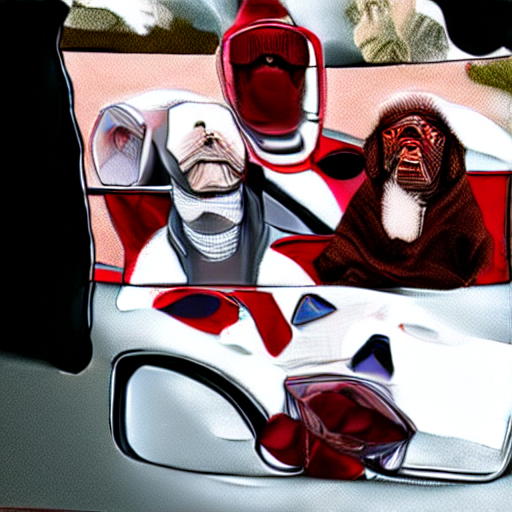} &
\includegraphics[width=0.13\linewidth]{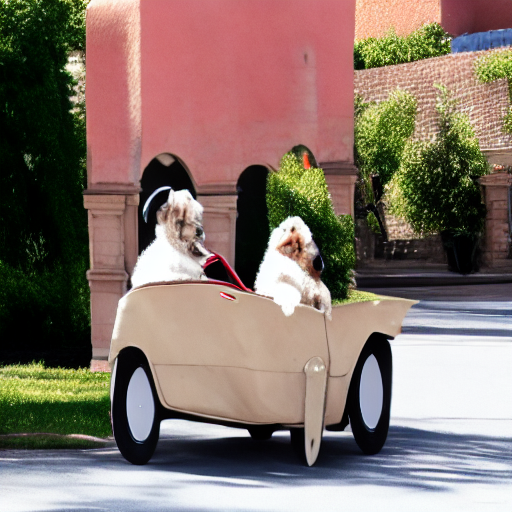} &
\includegraphics[width=0.13\linewidth]{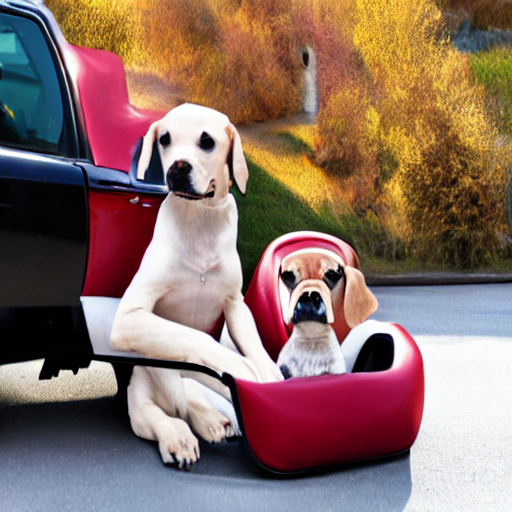} \\
\hdashline
\multirow{2}{*}{\raisebox{-5.5\height}{4}} &
$\makecell{0.118\ \scriptstyle(\pm 0.034)}$ &
$\makecell{0.137\ \scriptstyle(\pm 0.035)}$ &
$\makecell{0.159\ \scriptstyle(\pm 0.032)}$ &
$\makecell{0.211\ \scriptstyle(\pm 0.040)}$ &
$\makecell{0.231\ \scriptstyle(\pm 0.034)}$ \\
&
\includegraphics[width=0.13\linewidth]{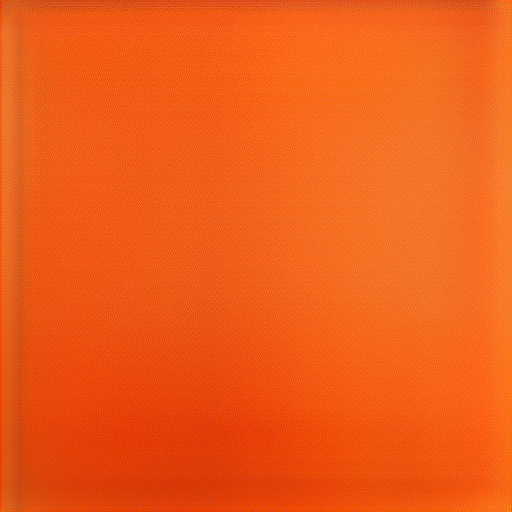} &
\includegraphics[width=0.13\linewidth]{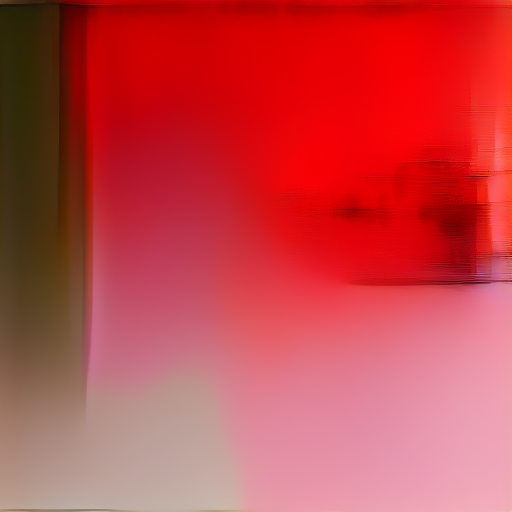} &
\includegraphics[width=0.13\linewidth]{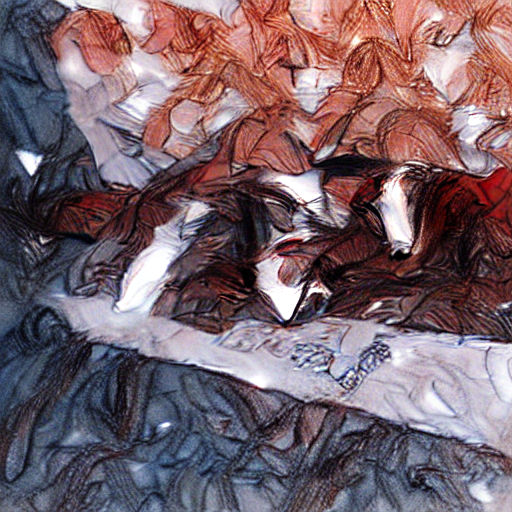} &
\includegraphics[width=0.13\linewidth]{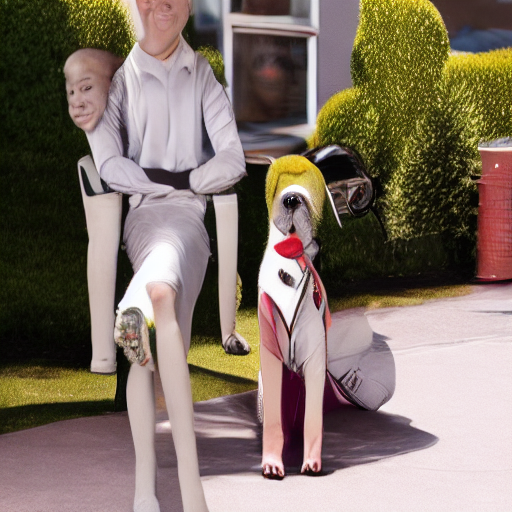} &
\includegraphics[width=0.13\linewidth]{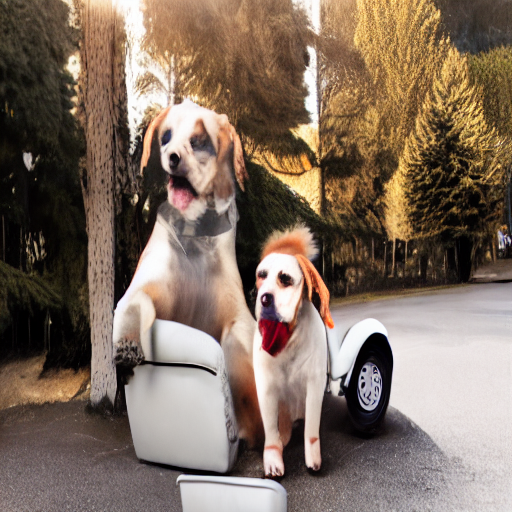} \\
\hdashline
\multirow{2}{*}{\raisebox{-5.5\height}{16}} &
$\makecell{0.127\ \scriptstyle(\pm 0.035)}$ &
$\makecell{0.129\ \scriptstyle(\pm 0.032)}$ &
$\makecell{0.147\ \scriptstyle(\pm 0.033)}$ &
$\makecell{0.201\ \scriptstyle(\pm 0.034)}$ &
$\makecell{0.220\ \scriptstyle(\pm 0.030)}$ \\
 &
\includegraphics[width=0.13\linewidth]{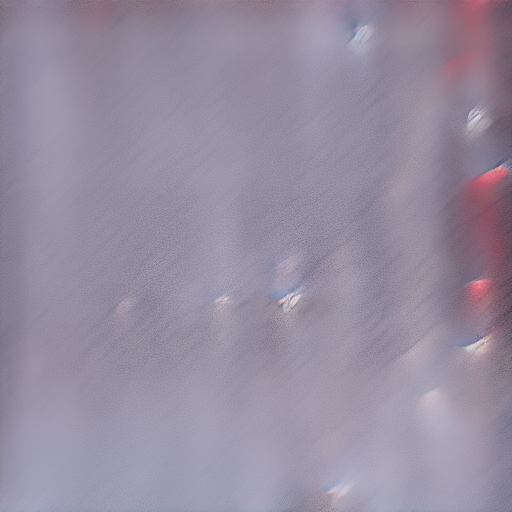} &
\includegraphics[width=0.13\linewidth]{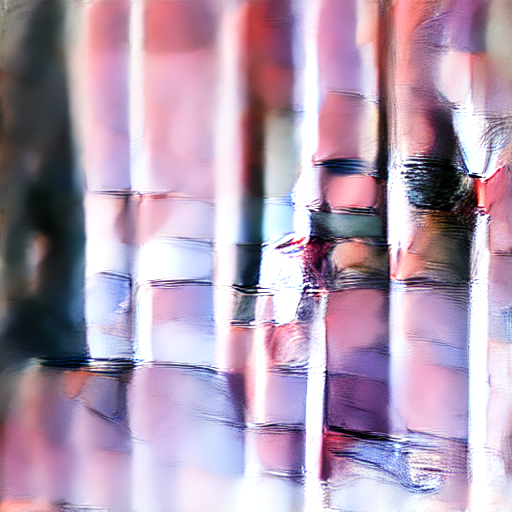} &
\includegraphics[width=0.13\linewidth]{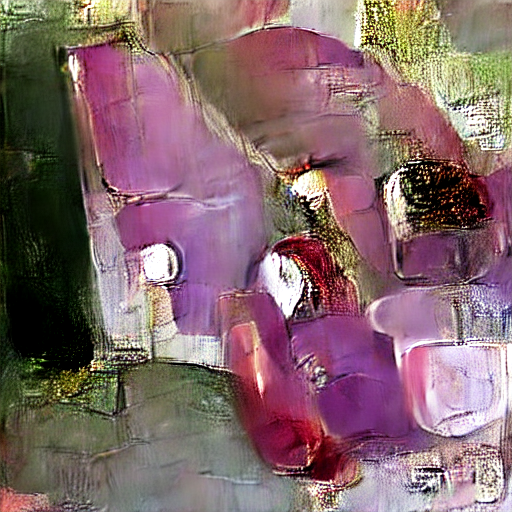} &
\includegraphics[width=0.13\linewidth]{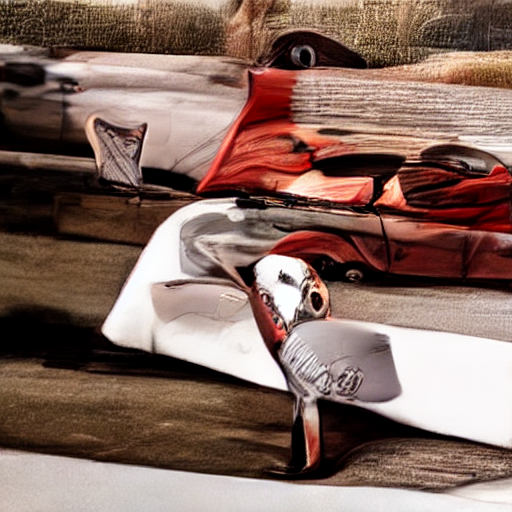} &
\includegraphics[width=0.13\linewidth]{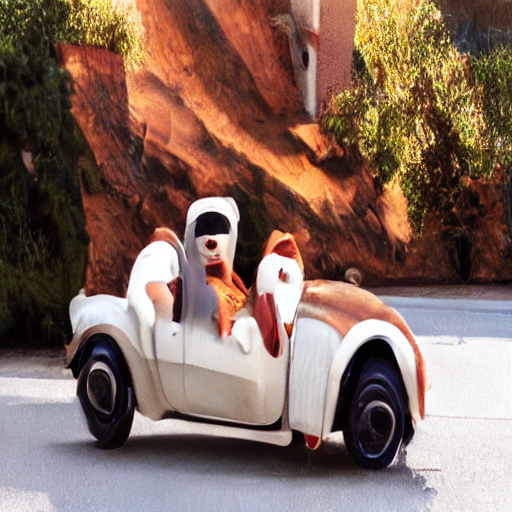} \\
\hdashline
\multirow{2}{*}{\raisebox{-5.5\height}{64}} &
$\makecell{0.132\ \scriptstyle(\pm 0.029)}$ &
$\makecell{0.125\ \scriptstyle(\pm 0.031)}$ &
$\makecell{0.135\ \scriptstyle(\pm 0.028)}$ &
$\makecell{0.169\ \scriptstyle(\pm 0.033)}$ &
$\makecell{0.187\ \scriptstyle(\pm 0.036)}$ \\
 &
\includegraphics[width=0.13\linewidth]{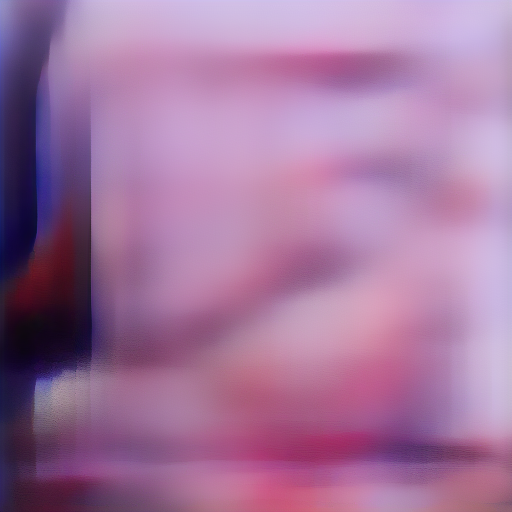} &
\includegraphics[width=0.13\linewidth]{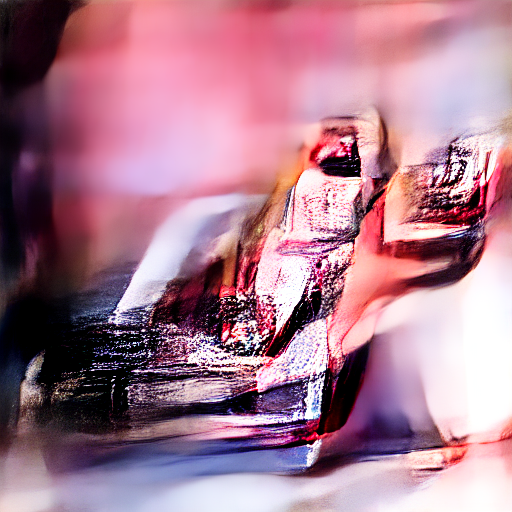} &
\includegraphics[width=0.13\linewidth]{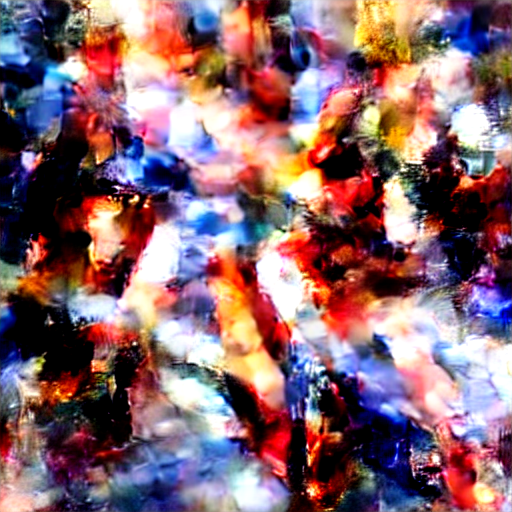} &
\includegraphics[width=0.13\linewidth]{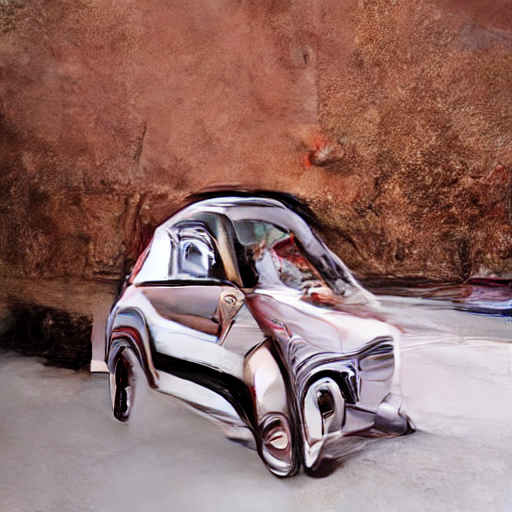} &
\includegraphics[width=0.13\linewidth]{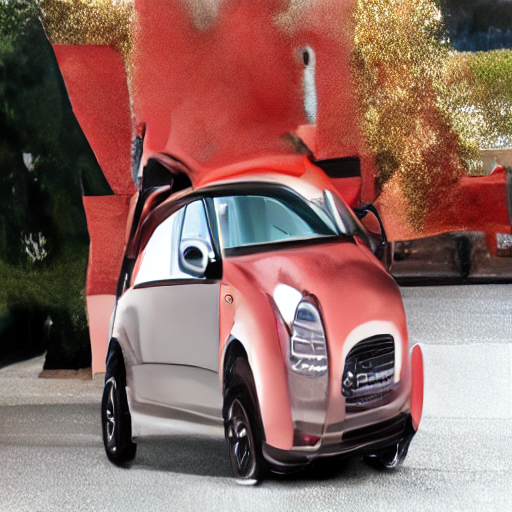} \\
\hline
\multirow{2}{*}{\raisebox{-5.5\height}{Baseline}} &
$\makecell{0.267\ \scriptstyle(\pm 0.032)}$ &
 &
 &
 &
 \\
 & \includegraphics[width=0.13\linewidth]{fig/images_LoREnc_example/0_-1.png} & & & \\
\hline
\end{tabular}
\end{center}
\end{table*}

\section{Pseudo-code of LoREnc} \label{appendix:pseudocode}
Algorithm~\ref{alg:lorenc} presents Python-style pseudocode for the proposed LoREnc framework.

\begin{algorithm*}[!t]
   \caption{Python-style pseudocode for LoREnc}
   
   \label{alg:lorenc}
   
    \definecolor{codeblue}{rgb}{0.25, 0.5, 0.5}
    \definecolor{codekw}{rgb}{0.85, 0.18, 0.50}
    \lstset{
  backgroundcolor=\color{white},
  basicstyle=\fontsize{7.1pt}{7.1pt}\ttfamily\selectfont,
  columns=fullflexible,
  breaklines=true,
  captionpos=b,
  commentstyle=\fontsize{7.1pt}{7.1pt}\color{codeblue},
  keywordstyle=\fontsize{7.1pt}{7.1pt}\color{codekw},
}
\begin{lstlisting}[language=python]
# Inputs:
#   W: Original FM weight matrix
#   {A_k, B_k}: Set of K downstream LoRA adapters
#   delta_r: Rank for Spectral Truncation
#   r: Rank of original LoRA adapters

#######################################
### Phase 1: Spectral Truncation    ###
#######################################

# 1. Extract the Spectral Key (L) via TSVD
# Note: kept _FM notation to indicate source
U_FM, S_FM, Vh_FM = TSVD(W, rank=delta_r)

# L represents the dominant low-rank structure of FM
L = U_FM @ diag(S_FM) @ Vh_FM

# 2. Truncate FM weights (Obfuscation)
W_tilde = W - L  # Deploy W_tilde to edge device


#######################################
### Phase 2: Spectral Compensation  ###
#######################################

for k in range(K):
    # Inject spectral components into adapters (Rank Expansion)
    # B_k: [m, r] -> [m, r + delta_r]
    # A_k: [r, n] -> [r + delta_r, n]
    
    # Distribute singular values symmetrically
    S_sqrt = diag(S_FM) ** 0.5
    
    # Concatenate FM spectral components to LoRA factors
    # B_comp: [B_k | U_FM * S_sqrt]
    B_comp[k] = concat([B[k], U_FM @ S_sqrt], axis=1)
    
    # A_comp: [A_k / (S_sqrt * Vh_FM)] (stacked vertically)
    A_comp[k] = concat([A[k], S_sqrt @ Vh_FM], axis=0)


#######################################
### Phase 3: Secure Adapter Encoding ###
#######################################

for k in range(K):
    # 1. Generate Restoration Keys via SVD on compensated adapter
    #    Note: Use _Lo (or _Enc) to distinguish from FM components
    U_Lo, S_Lo, Vh_Lo = SVD(B_comp[k] @ A_comp[k])

    # 2. Split into Restoration Key (delta_r) and Encrypted Adapter (r)
    S_Lo_sqrt = diag(S_Lo) ** 0.5
    
    # K_B: Restoration Key, B_enc: Encrypted LoRA B
    K_B[k], B_enc[k] = split(
        U_Lo @ S_Lo_sqrt, 
        split_sizes=[delta_r, r], axis=1
    )
    
    # K_A: Restoration Key, A_enc: Encrypted LoRA A
    K_A[k], A_enc[k] = split(
        S_Lo_sqrt @ Vh_Lo, 
        split_sizes=[delta_r, r], axis=0
    )

    # 3. Orthogonal Reparameterization
    M = random_orthogonal_matrix(r)  # size: r x r
    
    B_final[k] = B_enc[k] @ M        # Encrypted B
    A_final[k] = M.T @ A_enc[k]      # Encrypted A

# Output:
#   W_tilde: Truncated FM
#   {B_final, A_final}: Encrypted Adapters
#   {K_B, K_A}: Restoration Keys
return W_tilde, (B_final, A_final), (K_B, K_A)

\end{lstlisting}
\end{algorithm*}

\end{document}